\documentclass[reqno]{amsart}
\usepackage{amsmath}
\usepackage{amsfonts}
\usepackage{amssymb}
\usepackage{mathrsfs}
\usepackage{xcolor}
\usepackage{graphicx}
\usepackage[colorlinks=true, allcolors=blue]{hyperref}
\usepackage[font=footnotesize,labelfont=bf]{caption}
\usepackage{bm}
\usepackage{pgfplots}
\pgfplotsset{compat=1.18}
\usepackage[noadjust]{cite}

\newtheorem{theorem}{Theorem}[section]
\newtheorem{corollary}[theorem]{Corollary}
\newtheorem{lemma}[theorem]{Lemma}
\newtheorem{proposition}[theorem]{Proposition}
\newtheorem{problem}[theorem]{Problem}

\theoremstyle{remark}
\newtheorem{definition}[theorem]{Definition}
\newtheorem{remark}[theorem]{Remark}
\newtheorem{example}[theorem]{Example}

\newcommand{\fei}[3]{$\alpha$-$\operatorname{EI}_\mathbb{F}(#1, #2, #3)$}

\begin{document}

\title[Dimension-counting bounds for equi-isoclinic subspaces]{Dimension-counting bounds for \\ equi-isoclinic subspaces}
\author{Joseph W.\ Iverson
\qquad
Kaysie Rose O.}
\address{Department of Mathematics, Iowa State University, Ames, IA}

\email{jwi@iastate.edu}
\email{kaysie@iastate.edu}

\begin{abstract}
We make four contributions to the theory of optimal subspace packings and equi-isoclinic subspaces: 
(1) a new lower bound for block coherence, (2) an exact count of equi-isoclinic subspaces of even dimension $r$ in $\mathbb{R}^{2r+1}$ with parameter $\alpha \neq \tfrac{1}{2}$, (3) a new upper bound for the number of $r$-dimensional equi-isoclinic subspaces in $\mathbb{R}^d$ or $\mathbb{C}^d$, and (4) a proof that when $d=2r$, a further refinement of this bound is attained for every $r$ in the complex case and every $r=2^k$ in the real case.
For each of these contributions, the proof ultimately relies on a dimension count.
\end{abstract}

\maketitle

\section{Introduction}

In a finite-dimensional Hilbert space $\mathcal{H}$, the \textbf{first principal angle} between nontrivial subspaces $V,W \leq \mathcal{H}$ is defined as the smallest angle between lines in the two spaces, namely,
\[
\theta_1(V,W) := \arccos \max\{ | \langle v,w \rangle | : v \in V,\, w \in W, \, \| v \| = \| w \| = 1\}.
\]
First considered by Jordan~\cite{Jordan:75}, this notion (or more precisely, its sine) corresponds to the distance between $V$ and $W$ when they are viewed as sets in projective space.
For a given field $\mathbb{F} \in \{ \mathbb{R}, \mathbb{C} \}$ and parameters $d,r \geq 1$ and $n \geq 2$, the \textit{spectral packing problem}~\cite{DhillonHST:08,FickusJMW:17} asks to find a sequence $\{W_j\}_{j=1}^n$ of $r$-dimensional subspaces 
$W_j \leq \mathcal{H}$ in a $d$-dimensional Hilbert space over $\mathbb{F}$ for which the smallest principal angle $\min_{i \neq j} \theta_1(W_i,W_j)$ is as large as possible.
Equivalently, the \textbf{block coherence}
\[
\mu := \cos \min_{i \neq j} \theta_1(W_i,W_j) 
= \max_{i \neq j} \cos \theta_1(W_i,W_j)
\]
should be as small as possible.
This terminology comes from 
compressed sensing~\cite{Tropp:04,EldarKB:10,CalderbankTX:15,EldarM:09}, where subspaces with small block coherence 
allow recovery of block-sparse signals from few linear measurements.
(The smaller the block coherence, the more nonzero blocks are allowed in the signal.)
Motivated by such applications, we seek: (1) lower bounds on block coherence, and (2) theoretical limits on subspaces that achieve those bounds.

For example, the \textbf{Welch bound}~\cite{Welch:74,DhillonHST:08,CalderbankTX:15,FickusJMW:17} states that
\begin{equation}
\label{eq:Welch}
\mu \geq \sqrt{\frac{\tfrac{n}{d/r}-1}{n-1}}.
\end{equation}
Conditions for equality can be understood in terms of the orthogonal projections $P_1,\ldots,P_n$ for $W_1,\ldots,W_n$, where the Welch bound is attained if and only if:
    \begin{itemize}
    \item[(i)]
    $\exists \alpha$ with $\| P_i x \| = \alpha \| x \|$ for every $x \in W_j$ and $i \neq j$, and
    \smallskip
    \item[(ii)]
    $\exists c$ with $P_1 + \dotsb + P_n = cI$.
\end{itemize}
If (i) holds, then the subspaces are \textbf{equi-isoclinic} with parameter $\alpha \in [0,1]$, and we call $\{W_j\}_{j=1}^n$ an $\alpha$-$\operatorname{EI}_{\mathbb{F}}(d,r,n)$.\footnote{Some texts (like~\cite{LemmensS:73}) call $\alpha^2$ the parameter instead.}
If (ii) holds, then we call $\{W_j\}_{j=1}^n$ a \textbf{tight fusion frame}, abbreviated $\operatorname{TFF}_{\mathbb{F}}(d,r,n)$.
If both (i) and (ii) hold (that is, if the Welch bound is attained), then we call $\{W_j\}_{j=1}^n$ an \textbf{equi-isoclinic tight fusion frame}, abbreviated $\operatorname{EITFF}_{\mathbb{F}}(d,r,n)$.

Equi-isoclinic subspaces can be viewed as higher-dimensional analogs of equiangular lines, the study of which Godsil and Royle called ``one of the founding problems of algebraic graph theory''~\cite{GodsilR:01}.
Accordingly, equi-isoclinic subspaces have been studied by experts in the algebraic combinatorics, data science, and harmonic analysis communities~\cite{CalderbankTX:15,DhillonHST:08,EldarKB:10,EtTaoui:06,EtTaoui:18,EtTaoui:20,FGLI:25,FIJM:24,FickusIJM:23,FickusIJM:24,FickusJMW:17,GodsilH:92,Hoggar:76,Hoggar:77,IversonKM:21,LemmensS:73,Waldron:20}.
The main problem in this area is to determine parameters for which equi-isoclinic subspaces exist, and in particular, to determine the maximum number $v_{\mathbb{F}}(r,d)$ of distinct equi-isoclinic $r$-dimensional subspaces in $\mathbb{F}^d$.
This question is relevant for the spectral packing problem since its answer provides a necessary condition for Welch bound equality.

In this paper, we make four contributions to the spectral packing problem and the study of equi-isoclinic subspaces.
First, we provide a new lower bound on block coherence (Theorem~\ref{thm: spark bound for coherence}).
Second, we determine the existence of $\alpha$-$\operatorname{EI}_{\mathbb{R}}(2r+1,r,n)$ when $r$ is even and $\alpha \neq \tfrac{1}{2}$ (Theorem~\ref{thm:d2r1}).
Third, we provide a new upper bound on $v_{\mathbb{F}}(r,d)$ (Theorem~\ref{thm:biggerzonbound}).
Fourth, we show a refinement of this upper bound is achieved infinitely often (Theorem~\ref{thm:RHdimKn=n}).

All these contributions are joined by a common theme: in each case, the proof boils down to a non-obvious dimension count.
For example, our new upper bound on $v_{\mathbb{F}}(r,d)$ follows a strategy of Gerzon reported by Lemmens and Seidel in 1973~\cite{LemmensS:73}.
Gerzon observed that when $\alpha \neq 1$, the orthogonal projections of an $\alpha$-$\operatorname{EI}_{\mathbb{F}}(d,r,n)$ are linearly independent (Proposition~\ref{prop: independence}), so $n$ is at most the dimension of $\mathbb{F}_H^{d\times d} := \{ M \in \mathbb{F}^{d \times d} : M^* = M\}$, namely, $\frac{d(d+1)}{2}$ when $\mathbb{F}=\mathbb{R}$ and $d^2$ when $\mathbb{F}=\mathbb{C}$.
Lemmens and Seidel further noted these projections reside in a certain subspace $\mathcal{K}_1 \leq \mathbb{F}_H^{d\times d}$ whose dimension is easy to compute; this gives the bound
\[
n \leq \dim \mathcal{K}_1 = \dim \mathbb{F}_H^{d\times d} - \dim \mathbb{F}_H^{r \times r}+1.
\]
Refining this technique, we identify a nesting family of 
subspaces
\[
\mathcal{K}_n \leq \mathcal{K}_{n-1} \leq \cdots \leq \mathcal{K}_1,
\]
each of which contains all the projections.
This gives a sequence of $n$ inequalities $n \leq \dim \mathcal{K}_j$ for $j \leq n$.
Unfortunately, it is harder to compute $\dim \mathcal{K}_j$ when $j>1$; fortunately, we accomplish it for $j=2$ and $j=3$, finding (Theorem~\ref{thm: dim K3})
\begin{equation}
\label{eq: K3 bound}
n \leq \dim \mathcal{K}_3 = \dim \mathbb{F}_H^{d\times d} - 3\dim \mathbb{F}_H^{r \times r}+3.
\end{equation}
This improves Lemmens and Seidel's 1973 bound when $r>1$.

Furthermore, we show the final inequality $n \leq \dim \mathcal{K}_n$ is achieved with $n = \dim \mathcal{K}_n = v_{\mathbb{F}}(r,2r)$ infinitely often in the case $d = 2r$: it happens for every $r$ when $\mathbb{F} = \mathbb{C}$ and for every $r=2^k$ when $\mathbb{F} = \mathbb{R}$.
Since the inequality $n \leq \dim \mathcal{K}_n$ is a refinement of Gerzon's bound, this last result can be viewed as an analog of \textit{Zauner's conjecture}~\cite{Zauner:11} for higher-dimensional subspaces.

Our upper bound~\eqref{eq: K3 bound} makes progress on a question of Balla and Sudakov~\cite{BallaS:19}, who asked to determine the asymptotic dependence of $v_{\mathbb{F}}(r,d)$ on $r$ as $r \to \infty$.
This is a subtle question since there are many ways for $d$ to grow along with $r$.
For example, when $d=2r$, it is known that $v_{\mathbb{F}}(r,2r)=\rho_{\mathbb{F}}(r)+2$, where $\rho_{\mathbb{F}}(r)$ denotes the Radon--Hurwitz number~\cite{LemmensS:73,Hoggar:76}.
The value of $\rho_{\mathbb{F}}(r)$ is given in~\eqref{eq: radon hurwitz number values} and depends on the largest power of 2 that divides~$r$; hence $\rho_{\mathbb{F}}(r) = O(\log r)$.
When $d=mr$ for some other fixed $m$, very little is known.
Lemmens and Seidel's bound states 
\[
v_{\mathbb{F}}(r,mr) 
\leq \dim \mathcal{K}_1 
= \begin{cases}
\frac{m^2-1}{2} r^2 + \frac{m-1}{2} r + 1 & \mathbb{F} = \mathbb{R}, \\[5 pt]
(m^2-1)r^2 + 1 & \mathbb{F} = \mathbb{C};
\end{cases}
\]
\eqref{eq: K3 bound} improves this to 
\[
v_{\mathbb{F}}(r,mr) 
\leq \dim \mathcal{K}_3 
= \begin{cases}
\frac{m^2-3}{2} r^2 + \frac{m-3}{2} r + 3 & \mathbb{F} = \mathbb{R}, \\[5 pt]
(m^2-3)r^2 + 3 & \mathbb{F} = \mathbb{C},
\end{cases}
\]
decreasing the leading coefficient.
While this bound remains far from the true value of $v_{\mathbb{F}}(r,mr)$ when $m=2$, we find it promising that $n = \dim \mathcal{K}_n=v_{\mathbb{F}}(r,2r)$ infinitely often even in this case.

\section{Preliminaries}

For the remainder of the paper, we focus on the Hilbert space $\mathcal{H} = \mathbb{F}^d$.
We begin with a short review of subspace packings and fusion frames.

\subsection{Principal angles and isoclinism}

Given a pair of $r$-dimensional subspaces $V,W \leq \mathbb{F}^d$, select orthonormal bases to form the columns of corresponding isometries $\Phi,\Psi \in \mathbb{F}^{d \times r}$.
Denoting $P=\Phi \Phi^*$ for orthogonal projection onto $V$, the first principal angle between $V$ and $W$ is given by
\[
\cos \theta_1(V,W) = \max\{ \| P w \| : w \in W, \, \| w \| = 1 \} = \| \Psi^* \Phi \|_{\text{op}},
\]
where $\| \cdot \|_{\text{op}} = \sigma_{\text{max}}(\cdot)$ denotes the operator norm.
Indeed, the first equation holds since any unit vectors $v \in V$ and $w \in W$ satisfy $| \langle v, w \rangle | = |\langle v, Pw \rangle | \leq \| Pw \|$, with equality when $Pw \in \operatorname{span}\{v\}$, and the second equation holds since $\Phi^* P \Psi = \Phi^* \Psi$ is the matrix of the transformation $W \to V \colon x \mapsto Px$ relative to the chosen bases.
More generally, for $j \in \{1,\ldots,r\}$, the $j$th \textbf{principal angle} between $V$ and $W$ is defined as $\theta_j(V,W) := \arccos \sigma_j(\Phi^* \Psi)$.
It does not depend on the choice of bases since $(\Phi U)^*(\Psi U')= U^*(\Phi^* \Psi)U'$ and $\Phi^* \Psi$ have the same singular values when $U,U' \in \mathbb{F}^{r \times r}$ are unitary.
Furthermore, $\theta_j(W,V) = \theta_j(V,W)$ since $\Psi^* \Phi$ and $\Phi^* \Psi$ have the same singular values.
It is easy to prove the following, which summarizes Theorem~2.3 of~\cite{LemmensS:73}.

\begin{lemma}[\cite{LemmensS:73}]
\label{lem: isoclinism}
Given $r$-dimensional subspaces $V,W \in \mathbb{F}^d$, select corresponding isometries $\Phi,\Psi \in \mathbb{F}^{d \times r}$, and let $P,Q \in \mathbb{F}^{d \times d}$ be the corresponding orthogonal projections.
Then the following are equivalent for any choice of $\alpha \in [0,1]$:
    \begin{itemize}
    \item[(i)]
    $\cos \theta_j(V,W) = \alpha$ for every $j \in [r]$,
    \smallskip
    \item[(ii)]
    $\| P x \| = \alpha \| x \|$ for every $x \in W$,
    \smallskip
    \item[(iii)]
    $\Phi^* \Psi = \alpha U$ for some unitary $U \in \mathbb{F}^{r \times r}$,
    \smallskip
    \item[(iv)]
    $\Psi^* P \Psi = \alpha^2 I$,
    \smallskip
    \item[(v)]
    $QPQ = \alpha^2 Q$.
    \end{itemize}
\end{lemma}

When the equivalent conditions above hold, we say $V$ and $W$ are $\alpha$-\textbf{isoclinic}.
Evidently, a sequence $\{W_j\}_{j=1}^n$ of $r$-dimensional subspaces $W_j \leq \mathbb{F}^d$ is equi-isoclinic with parameter $\alpha$ if and only if $W_i$ and $W_j$ are $\alpha$-isoclinic whenever $i \neq j$.
Furthermore, any $\alpha$-$\operatorname{EI}$ has block coherence $\mu = \alpha$.
In particular, an $\alpha$-$\operatorname{EI}_{\mathbb{F}}(d,r,n)$ is an EITFF if and only if
\begin{equation}
\label{eq:EITFF alpha}
\alpha = \sqrt{\frac{\tfrac{n}{d/r}-1}{n-1}}.
\end{equation}
Notably, if $d<2r$ then any $\alpha$-$\operatorname{EI}(d,r,n)$ has $\alpha=1$ since any pair of $r$-dimensional subspaces in $\mathbb{F}^d$ intersect nontrivially, for block coherence $\mu = 1$.
Considering~\eqref{eq:EITFF alpha}, an $\operatorname{EITFF}(d,r,n)$ exists only if $d=r$ or $d \geq 2r$.

\begin{remark}
This paper relies on the first principal angle to measure distance between subspaces, but there are many other such notions, all of which involve the principal angles.
See~\cite{ConwayHS:96,DhillonHST:08} for a summary and discussion of the related subspace packing problems.

Likewise, there are other ways to generalize the notion of equiangular lines to higher-dimensional subspaces, besides equi-isoclinic subspaces.
Examples include \textit{equi-chordal} subspaces~\cite{FickusJMW:17} and (much more generally) \textit{equiangular subspaces}~\cite{Blokhuis:93,BallaDKS:17,BallaS:19,WangY:25}.
\end{remark}

\subsection{Fusion frames}
Let $\{W_j\}_{j=1}^n$ be a sequence of $r$-dimensional subspaces $W_j \leq \mathbb{F}^d$, and let $P_j \in \mathbb{F}^{d \times d}$ be orthogonal projection onto $W_j$.
We call $\{W_j\}_{j=1}^n$ a \textbf{fusion frame} if $W_1+\dotsb+W_n = \mathbb{F}^d$.
Recall that $\{W_j\}_{j=1}^n$ is a \textit{tight fusion frame} if $P_1+\dotsb+P_n = c I$ for some $c$, necessarily $c = \tfrac{rn}{d}$, as seen by taking a trace.
For another perspective, select an isometry $\Phi_j \in \mathbb{F}^{d \times r}$ with image $W_j$ for each $j$, so that $P_j = \Phi_j \Phi_j^*$, and consider the \textbf{fusion synthesis matrix}
\[
\Phi := \begin{bmatrix} \Phi_1 & \cdots & \Phi_n \end{bmatrix} \in (\mathbb{F}^{d \times r})^{1 \times n} \cong \mathbb{F}^{d \times rn}.
\]
Then $P_1 + \dotsb + P_n$ coincides with the \textbf{fusion frame operator} $\Phi \Phi^*$, which shares its nonzero eigenvalues with the \textbf{fusion Gram matrix} $\Phi^* \Phi = \left[ \Phi_i^* \Phi_j \right]_{i,j=1}^n$.
In particular, $\{W_j\}_{j=1}^n$ is a tight fusion frame if and only if $\tfrac{d}{rn} \Phi^* \Phi$ is an orthogonal projection of rank~$d$.
In that case, if $d \neq rn$ then $I - \tfrac{d}{rn} \Phi^* \Phi$ is an orthogonal projection of rank~$rn-d$, and it follows that $\tfrac{rn}{rn-d}(I - \tfrac{d}{rn} \Phi^* \Phi)$ is the fusion Gram matrix of some $\operatorname{TFF}_{\mathbb{F}}(rn-d,r,n)$.
Any such TFF is called a \textbf{Naimark complement} for the TFF $\{W_j\}_{j=1}^n$.
Furthermore, if $\{W_j\}_{j=1}^n$ is an $\operatorname{EITFF}_{\mathbb{F}}(d,r,n)$ with isoclinism parameter $\alpha$, then each off-diagonal block of $\Phi^*\Phi$ is $\alpha$ times a unitary by Lemma~\ref{lem: isoclinism}(iii); if $d \neq rn$ then a similar consideration shows any Naimark complement is an $\operatorname{EITFF}_{\mathbb{F}}(rn-d,r,n)$.

\begin{remark}
While we were writing Section~\ref{sec: lower bound}, the ratio $d/r\geq 1$ emerged as a fundamental parameter, and we found it possible to write all the hypotheses and conclusions in that section in terms of the pair $(d/r,n)$ only, without need for the triple $(d,r,n)$.
We find it remarkable that such a reduction was possible.
In fact, this phenomenon extends to much of the basic theory.
Below, we mention several instances for a sequence of $r$-dimensional subspaces $W_1,\ldots,W_n \leq \mathbb{F}^d$:

    \begin{itemize}
    \item[(a)]
    $W_j = \mathbb{F}^d$ if and only if $d/r = 1$,
    \smallskip
    \item[(b)]
    $W_1+\dotsb+W_n = \mathbb{F}^d$ only if $d/r \leq n$,
    \smallskip
    \item[(c)]
    $W_1 \oplus \dotsb \oplus W_n = \mathbb{F}^d$ only if $d/r = n$,
   \smallskip
    \item[(d)]
    the Naimark complement of a $\operatorname{TFF}_{\mathbb{F}}(d,r,n)$ is a $\operatorname{TFF}_{\mathbb{F}}(D,r,n)$, where $D/r=n-d/r$,
    \smallskip
    \item[(e)]
    $\{W_j\}_{j=1}^n$ is the Naimark complement of a $\operatorname{TFF}_{\mathbb{F}}(r,r,n)$ only if $d/r = n-1$,
    
    \smallskip
    \item[(f)]
    apart from the trivial cases $d/r \in \{1,n-1,n\}$, an $\operatorname{EITFF}(d,r,n)$ exists only if $d/r \in [2,n-2]$,
    \smallskip
    \item[(g)]
    given an $\operatorname{EITFF}_{\mathbb{F}}(d,r,n)$ and an $\operatorname{EITFF}_{\mathbb{F}}(d',r',n)$ with $d/r = d'/r'$, a ``direct sum'' procedure yields an $\operatorname{EITFF}_{\mathbb{F}}(d+d',r+r',n)$, and this preserves the ratio $(d+d')/(r+r')=d/r=d'/r'$ (see~\cite{FickusMW:21}),
    \smallskip
    \item[(h)]
    when $\mathbb{F}=\mathbb{C}$, Hoggar's ``$\mathbb{C}$-to-$\mathbb{R}$ trick'' converts any $\operatorname{EITFF}_{\mathbb{C}}(d,r,n)$ into an $\operatorname{EITFF}_{\mathbb{R}}(2d,2r,n)$, and this preserves the ratio $(2d)/(2r)=d/r$ (see~\cite{Hoggar:77}),
    \smallskip
    \item[(i)]
    the Welch bound~\eqref{eq:Welch} is a function of $(d/r,n)$, and
    \smallskip
    \item[(j)]
    the ``spark bound for coherence'' (Theorem~\ref{thm: spark bound for coherence} below) is a function of $d/r$.
    \end{itemize}
Despite these last two facts, the optimal block coherence cannot always be expressed as a function of $(d/r,n)$.
There are many examples where a $d \times n$ complex \textit{equiangular tight frame} (ETF) is known to exist, but a $d \times n$ real ETF provably does not exist~\cite{FickusM:15}.
For any such example, Hoggar's ``$\mathbb{C}$-to-$\mathbb{R}$'' trick provides an $\operatorname{EITFF}_{\mathbb{R}}(2d,2,n)$ that achieves the Welch bound, where no $\operatorname{EITFF}_{\mathbb{R}}(d,1,n)$ exists to achieve the same bound.
Similarly, an $\operatorname{EITFF}_{\mathbb{C}}(4,2,6)$ exists by virtue of Hoggar's ``$\mathbb{H}$-to-$\mathbb{C}$'' trick applied to a $2 \times 6$ quaternionic ETF~\cite{CohnKM:16}, but there is no $\operatorname{EITFF}_{\mathbb{C}}(2,1,6)$ by virtue of Gerzon's bound $n \leq d^2$.
\end{remark}

\section{A lower bound for block coherence}
\label{sec: lower bound}

When $1 < d/r < 2$, the Welch bound~\eqref{eq:Welch} severely underestimates the block coherence of $n \geq 2$ subspaces with dimension $r$ in $\mathbb{F}^d$.
Indeed, the Welch bound is strictly less than~1 in this setting, but any two such subspaces have nontrivial intersection since $d<2r$, and so the true block coherence is $\mu=1$.
In this section, we provide a new lower bound for block coherence that captures this simple dimension-counting argument, generalizing it for larger values of $d/r$.

\begin{theorem}[Spark bound for coherence]
\label{thm: spark bound for coherence}
Choose $d,r \geq 1$ and $n \geq 2$ with $1 < d/r < n$.
Then any sequence of $n$~subspaces with dimension~$r$ in $\mathbb{F}^d$ has block coherence
\begin{equation}
\label{eq: spark bound}
\mu \geq \frac{1}{\lfloor d/r \rfloor}.
\end{equation}
In particular, there does not exist an $\alpha$-$\operatorname{EI}_{\mathbb{F}}(d,r,n)$ for any $\alpha < \frac{1}{\lfloor d/r \rfloor}$.
\end{theorem}

A proof appears below.
Recall that a sequence $\{W_j\}_{j=1}^n$ of subspaces in $\mathbb{F}^d$ is called \textbf{linearly dependent} if there exist vectors $w_1,\ldots,w_n$, not all zero, such that $w_j \in W_j$ for every $j \in [n]$ and $w_1+\dotsb + w_n = 0$.
The \textbf{spark} of a linearly dependent sequence $\{W_j\}_{j=1}^n$ of subspaces in $\mathbb{F}^d$ is the smallest length of a linearly dependent subsequence.
We call Theorem~\ref{thm: spark bound for coherence} the ``spark bound for coherence'' since it follows quickly from the following ``coherence bound for spark'', which is well known for~${r=1}$~{\cite{FickusJKM:18,BandeiraFMW:13}}, and which appears implicitly in~\cite{EldarKB:10}.
Our contribution in this section is the observation that the ``coherence bound for spark'' sometimes yields a competitive lower bound for block coherence.

\begin{proposition}[Coherence bound for spark]
\label{prop: coherence bound for spark}
For $n \geq 2$, let $\{ W_j \}_{j=1}^n$ be a sequence of $r$-dimensional subspaces in $\mathbb{F}^d$ with block coherence $\mu \in (0,1]$.
Then any subsequence of $k \leq \min\{n,\lceil 1/\mu \rceil\}$ subspaces from $\{W_j\}_{j=1}^n$ is linearly independent.
Thus, if $\{W_j\}_{j=1}^n$ is linearly dependent, then its spark is at least $\lceil 1/\mu \rceil+1$.
\end{proposition}

We include the simple proof for the sake of clarity.

\begin{proof}[Proof of Proposition~\ref{prop: coherence bound for spark}]
Given any subsequence of $k \leq \min \{ n, \lceil 1/\mu \rceil \}$ subspaces from $\{W_j\}_{j=1}^n$, choose respective isometries $\Phi_1,\ldots,\Phi_k \in \mathbb{F}^{d \times r}$, and consider the fusion synthesis matrix
\[
\Phi := \begin{bmatrix} \Phi_1 & \cdots & \Phi_k \end{bmatrix} \in ( \mathbb{F}^{d \times r} )^{1 \times k} \cong \mathbb{F}^{d \times rk}.
\]
To show the subsequence is linearly independent, it suffices to show $\Phi$ has trivial kernel, or equivalently, $0 \notin \sigma(\Phi^* \Phi)$.
Toward this end, we apply Gershgorin's theorem for block matrices~\cite[Theorem~1.13.1]{Tretter:08}.
In the present setting, $\Phi^*\Phi$ is Hermitian with identity blocks $\Phi_i^* \Phi_i = I_r$ on the diagonal, so Gershgorin's theorem states that
\[
\sigma( \Phi^* \Phi ) \subseteq \bigcup_{j=1}^k D \biggl(1, \sum_{i \neq j} \| \Phi_i^* \Phi_j \|_{\text{op}} \biggr),
\]
where $D(a,R) \subseteq \mathbb{C}$ denotes the disk of radius $R$ centered at $a$.
Applying the definition of block coherence to enlarge the radius of each disk above, we find that
\[
\sigma( \Phi^* \Phi ) \subseteq D(1,(k-1)\mu).
\]
Here, $k-1 < 1/\mu$ since $k \leq \lceil 1/\mu \rceil$.
Therefore,
\[
\lambda_{\text{min}}(\Phi^* \Phi) \geq 1 - (k-1) \mu > 0.
\qedhere
\]
\end{proof}

\begin{proof}[Proof of Theorem~\ref{thm: spark bound for coherence}]
Let $\{W_j\}_{j=1}^n$ be a sequence of $r$-dimensional subspaces in $\mathbb{F}^d$, and let $\mu$ be its coherence.
Since $nr > d$, these subspaces are linearly dependent, and $\mu \neq 0$.
Next, set $k: = {\lfloor d/r \rfloor + 1}$, which is the smallest integer with $k > d/r$.
Then $n \geq k$, and since $kr > d$, every subsequence of $k$ subspaces from $\{W_j\}_{j=1}^n$ is linearly dependent.
In other words, 
\[
k \geq \operatorname{spark}(\{W_j\}_{j=1}^n)
\geq \lceil 1 /\mu \rceil+1,
\]
where we have applied Proposition~\ref{prop: coherence bound for spark}.
Consequently,
\[
\lfloor d/r \rfloor = k-1
\geq \lceil1/\mu\rceil
\geq 1/\mu.
\]
Rearranging gives~\eqref{eq: spark bound}.
\end{proof}

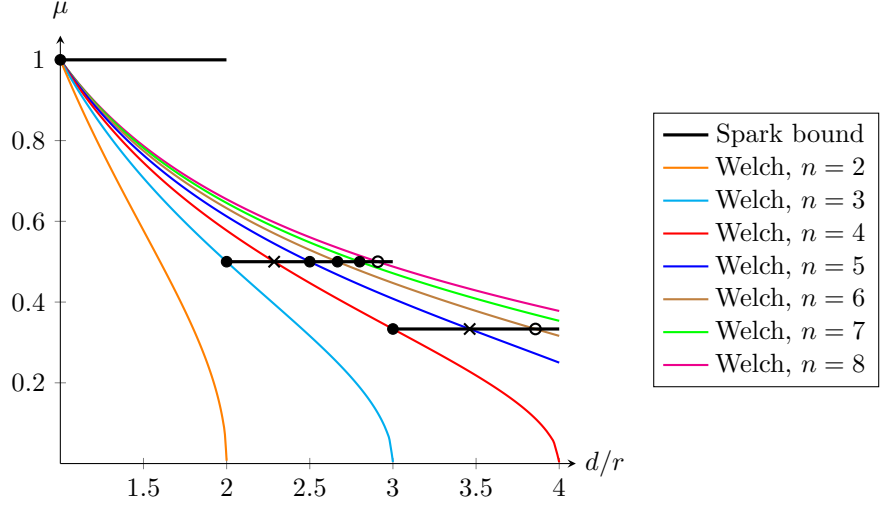
\begin{figure}
\begin{tikzpicture}
    \begin{axis}[
        y post scale=1,
        xlabel=$d/r$,
        ylabel=$\mu$,
        xmin=1, xmax=4.10,
        ymin=0, ymax=1.06,
        grid=none,
        xtick={1,1.5,2,2.5,3,3.5,4},
        ytick={0,0.2,0.4,0.6,0.8,1},
        axis lines=middle,
        legend style={at={(1.15,0.82)},anchor=north west},
        every axis x label/.style={at={(current axis.right of origin)},anchor=west},
        every axis y label/.style={at={(current axis.north west)},above=1.1mm},
        reverse legend
    ]

        \addplot[domain=1:4, samples=100, smooth, thick, magenta] {sqrt((8/x-1)/7)};
        \addlegendentry{Welch, $n=8$}

        \addplot[domain=1:4, samples=100, smooth, thick, green] {sqrt((7/x-1)/6)};
        \addlegendentry{Welch, $n=7$}

        \addplot[domain=1:4, samples=100, smooth, thick, brown] {sqrt((6/x-1)/5)};
        \addlegendentry{Welch, $n=6$}

        \addplot[domain=1:4, samples=100, smooth, thick, blue] {sqrt((5/x-1)/4)};
        \addlegendentry{Welch, $n=5$}

        \addplot[domain=1:4, samples=100, smooth, thick, red] {sqrt((4/x-1)/3)};
        \addlegendentry{Welch, $n=4$}

        \addplot[domain=1:3, samples=100, smooth, thick, cyan] {sqrt((3/x-1)/2)};
        \addlegendentry{Welch, $n=3$}

        \addplot[domain=1:2, samples=100, smooth, thick, orange] {sqrt((2/x-1)/1)};
        \addlegendentry{Welch, $n=2$}

        \addplot[domain=1:1.999, samples=100, smooth, very thick, black] {1/floor(x)};
        \addplot[domain=2:2.999, samples=100, smooth, very thick, black, forget plot] {1/floor(x)};
        \addplot[domain=3.001:4, samples=100, smooth, very thick, black, forget plot] {1/floor(x)};
        \addlegendentry{Spark bound}

        \addplot [only marks, forget plot] table {
            1 1
            2 0.5
            2.5 0.5
            2.6666667 0.5
            2.8 0.5
            3 0.3333333333
        };

        \addplot [only marks, mark=o, mark options={thick}, forget plot] table {
            2.90909090 0.5
            3.85714286 0.3333333333
        };

        \addplot [only marks, mark=x, mark options={scale=1.5,thick}, forget plot] table {
            2.28571429 0.5
            3.46153846 0.3333333333

        };
    \end{axis}
\end{tikzpicture}

{\footnotesize
\caption{ The \textit{spark bound} from Theorem~\ref{thm: spark bound for coherence} provides a lower bound for block coherence that exceeds the Welch bound~\eqref{eq:Welch} in some cases, as shown wherever a black line hovers over a colored line above. Where the spark and Welch bounds coincide, a point is marked. Filled circles represent known subspace packings, while open circles represent packings whose existence or nonexistence is unknown by the authors. At the points marked with an {\tt x}, no $\operatorname{EITFF}_{\mathbb{F}}(d,r,n)$ exists since $d/r > n-2$ and $d/r \notin\{n,n-1\}$.}
\label{fig: coherence bounds}
}
\end{figure}

\begin{example}
Given any dimensions $d$ and $r$ with $2 < d/r < 3$ and any $\alpha \in [\tfrac{1}{2},1)$, there exists an $\alpha$-$\operatorname{EI}_{\mathbb{C}}(d,r,n)$ with $n = 3$; taking $\alpha = \tfrac{1}{2}$ achieves the spark bound $\mu \geq \tfrac{1}{2}$.
Indeed, Lemma~\ref{lem:normalization} inspires the following construction.
Since $\tfrac{1}{2} \leq \alpha < 1$, it is not hard to show that 
\[
x_1 := \frac{3 \alpha^2 - 1}{2 \alpha^3} 
\]
satisfies $-1 \leq x_1 < 1$, and so we can choose $c>0$ for which
\[
x_2:= \frac{3 \alpha^2 - 1}{2 \alpha^3} + c^2 \cdot \frac{1-\alpha^2}{2 \alpha^3}
\]
also lies in $[-1,1]$.
Next, define $y_j := \sqrt{1 - x_j^2}$ and $\lambda_j := x_j + y_j \mathrm{i}$ for $j=1,2$, so that
\[
U:= \begin{bmatrix}
\lambda_2 I_{d-2r} & 0 \\
0 & \lambda_1 I_{3r-d}
\end{bmatrix}
\in \mathbb{C}^{r \times r}
\]
is unitary.
Then put $Y:=\begin{bmatrix} c I_{d-2r} & 0 \end{bmatrix} \in \mathbb{C}^{d-2r \times r}$ and $\beta:=\sqrt{1-\alpha^2}$, and define
\[
\Phi_1 := \begin{bmatrix}
I_r \\
0_r \\
0_{d-2r,r}
\end{bmatrix},
\quad
\Phi_2 := \begin{bmatrix}
\alpha I_r \\
\beta I_r \\
0_{d-2r,r}
\end{bmatrix},
\quad
\Phi_3 := \begin{bmatrix}
\alpha I_r \\
\tfrac{\alpha}{\beta} U - \tfrac{\alpha^2}{\beta} I_r \\
Y
\end{bmatrix}.
\]
With these definitions, it is straightforward (if a little tedious) to show that each $\Phi_j \in \mathbb{C}^{d \times r}$ is an isometry, and that the corresponding subspaces are equi-isoclinic with parameter~$\alpha$.
\end{example}

While the ``coherence bound for spark'' is well known in the case $r=1$, the ``spark bound for coherence'' is not useful in that case since it reads $\mu \geq \frac{1}{d}$, and one can easily show the Welch bound~\eqref{eq:Welch} is at least as strong under the hypotheses of Theorem~\ref{thm: spark bound for coherence} (namely, $d < n$).
However, for subspaces of larger dimension, there are cases where the ``spark bound for coherence'' is more informative,
as detailed below and depicted in Figure~\ref{fig: coherence bounds}.

\begin{corollary}
\label{cor: EITFF nonexistence}
Choose $d,r \geq 1$ and $n \geq 2$ with $1 < d/r < n$, and consider the Welch and spark bounds given by~\eqref{eq:Welch} and~\eqref{eq: spark bound}, respectively.
The spark bound exceeds the Welch bound if and only if one of the following holds:
    \begin{itemize}
    \item[(i)]
    $1 < d/r < 2$,
    \smallskip
    \item[(ii)]
    $n = \lceil d/r \rceil$ and $d/r \notin \mathbb{Z}$,
    \smallskip
    \item[(iii)]
    $n = \lceil d/r \rceil + 1$ and $f + \tfrac{f^2-f}{f^2+f+1} < d/r$, where $f = \lfloor d/r \rfloor$,
    \smallskip
    \item[(iv)]
    $5 \leq n \leq 8$ and $\tfrac{4n}{n+3} < d/r < 3$, or
    \smallskip
    \item[(v)]
    $n = 6$ and $\tfrac{27}{7} < d/r < 4$.
    \end{itemize}
In particular, if any of the above holds then there does not exist an $\operatorname{EITFF}_{\mathbb{F}}(d,r,n)$.
\end{corollary}

To be clear, EITFF nonexistence is novel only in cases (iv) and (v), since the necessary condition $d/r \in \{1,n-1,n\}\cup[2,n-2]$
fails in cases (i)--(iii).
In those ``trivial'' cases, the spark bound for coherence improves on EITFF nonexistence by explicitly bounding the block coherence away from Welch.

\begin{proof}[Proof of Corollary~\ref{cor: EITFF nonexistence}]
Write $d/r = f + p$ with $f = \lfloor d/r \rfloor \geq 1$ and $p \in [0,1)$, and put $k := n-f \geq 1$.
It is easy to show the spark bound exceeds the Welch bound if and only if
\begin{equation}
\label{eq: p bound}
p > \frac{(f^2-f)k-(f^2-f)}{f^2+f+k-1}.
\end{equation}
This happens in each of (i)--(v), where (i) says $f=1$ and $p>0$, (ii) says $k=1$ and $p>0$, (iii) interprets~\eqref{eq: p bound} for $k=2$, (iv) interprets it for $f=2$ and $3 \leq k \leq 6$, and (v) interprets it for $f=3$ and $k=3$.

To show these are the only possibilities, we may assume $f \geq 2$ and $k \geq 3$.
Since $p<1$, \eqref{eq: p bound}~holds only if
\[
\frac{(f^2-f)k-(f^2-f)}{f^2+f+k-1} < 1,
\]
that is,
\[
k < \frac{2f^2-1}{f^2-f-1}.
\]
This fails when $f \geq 4$ since $k \geq 3$;
for $f \in \{2,3\}$, it limits $k$ to the values given in (iv) and (v), respectively.
\end{proof}

\begin{table}[b]
\begin{center}
\small
\begin{tabular}{crrrrrrrrrrrrrrrrrr} 
\multicolumn{1}{c|}{$\bm{d}$} & 
8 & 11 & 11 & 13 & 14 & 14 & 16 & 17 & 17 & 17 & 18 & 19 & 19 & 20 & 20 & 20 & 21 & 22\\ 
\multicolumn{1}{c|}{$\bm{r}$} & 3 & 4 & 4 & 5 & 5 & 5 & 6 & 6 & 6 & 6 & 7 & 7 & 7 & 7 & 7 & 7 & 8 & 8\\ 
\multicolumn{1}{c|}{$\bm{n}$} & 5 & 5 & 6 & 5 & 5 & 6 & 5 & 5 & 6 & 7 & 5 & 5 & 6 & 5 & 6 & 7 & 5 & 5 \\
~ \\
\multicolumn{1}{c|}{$\bm{d}$} & 22 & 23 & 23 & 23 & 23 & 24 & 25 & 25 & 26 & 26 & 26 & 26 & 27 & 27 & 28 & 28 & 28 & 29\\ 
\multicolumn{1}{c|}{$\bm{r}$} & 8 & 8 & 8 & 8 & 9 & 9 & 9 & 9 & 9 & 9 & 9 & 10 & 10 & 10 & 10 & 10 & 11 & 10\\ 
\multicolumn{1}{c|}{$\bm{n}$} & 6 & 5 & 6 & 7 & 5 & 5 & 5 & 6 & 5 & 6 & 7 & 5 & 5 & 6 & 5 & 6 & 5 & 5 \\
\end{tabular}
\end{center}
\caption{For each column above, Corollary~\ref{cor: EITFF nonexistence} gives the nonexistence of an $\operatorname{EITFF}_{\mathbb{F}}(d,r,n)$ for either $\mathbb{F}\in \{ \mathbb{R}, \mathbb{C} \}$.
Considering Naimark complements, no $\operatorname{EITFF}_{\mathbb{F}}(rn-d,r,n)$ exists, either.
}
\label{tbl: nonexistence}
\end{table}

\begin{example}
Table~\ref{tbl: nonexistence} gives some parameters of EITFFs whose nonexistence follows from case (iv) or (v) of Corollary~\ref{cor: EITFF nonexistence}.
As far as the authors know, all these results are new.
\end{example}

There are a few places where the Welch and spark bounds coincide, as seen in Figure~\ref{fig: coherence bounds}.
Curiously, EITFFs are known to exist at several of these intersection points; when this happens, the spark bound is sharp.

\begin{example}
Given any integers $r,m \geq 1$, put $d := mr$ and $n:=m+1$.
Then the Welch and spark bounds coincide, where
\[
\sqrt{ \frac{\tfrac{n}{d/r} -1}{n-1} } = \frac{1}{m} = \frac{1}{\lfloor d/r \rfloor}.
\]
Both bounds are saturated since there exists an $\operatorname{EITFF}_{\mathbb{F}}(mr,r,m+1)$, namely, the Naimark complement of an $\operatorname{EITFF}_{\mathbb{F}}(r,r,m+1)$ that consists of $m+1$ copies of the ambient space $\mathbb{F}^r$.
In Figure~\ref{fig: coherence bounds}, this construction accounts for the filled circles at coordinates $(1,1)$, $(2,\tfrac{1}{2})$, and $(3,\tfrac{1}{3})$.
\end{example}

\begin{example}
\label{ex: totally symmetric spark bound}
Figure~\ref{fig: coherence bounds} shows three other filled circles where the Welch and spark bounds coincide.
An $\operatorname{EITFF}_{\mathbb{R}}(5,2,5)$ appears in~\cite{EtTaoui:06}, an $\operatorname{EITFF}_{\mathbb{R}}(8,3,6)$ appears in~\cite{IversonKM:21}, and the Naimark complement of an $\operatorname{EITFF}_{\mathbb{R}}(14,5,7)$ appears in~\cite{FickusIJM:24}.
Interestingly, all three circles can also be filled by examples in~\cite{FIJM:24}, which gives a common construction based on symmetry.
\end{example}

There are four other intersection points marked in Figure~\ref{fig: coherence bounds}.
At the points marked {\tt x}, EITFF nonexistence is trivial.
We do not know if EITFFs exist for the other two points (marked with open circles), and we leave this as an open problem.

\begin{problem}
Does there exist an $\operatorname{EITFF}_{\mathbb{F}}(d,r,n)$ with either
    \begin{itemize}
    \item[(a)]
    $(d/r,n) = (\tfrac{32}{11},8)$ or
    \smallskip
    \item[(b)]
    $(d/r,n) = (\tfrac{27}{7},6)$?
    \end{itemize}
\end{problem}


\section{An exact count of equi-isoclinic subspaces}

In this section, we determine the existence of $\alpha$-$\operatorname{EI}_{\mathbb{R}}(2r+1,r,n)$ when $r$ is even and $\alpha \neq \tfrac{1}{2}$.
We show such subspaces reside in a space of dimension $2r$, so they can be viewed as an $\alpha$-$\operatorname{EI}_{\mathbb{R}}(2r,r,n)$.
Existence of the latter was resolved by Lemmens and Seidel in terms of the (real) \textbf{Radon--Hurwitz number}, defined as $\rho_{\mathbb{R}}(r)=8b+2^c$ when $r=(2a+1)2^{4b+c}$ for nonnegative integers $a,b,c$ with $c \leq 3$.

\begin{proposition}[\cite{LemmensS:73}]
\label{prop:d2r}
For any $r \geq 1$, $n \geq 2$, and $\alpha \in (0,1)$, an $\alpha$-$\operatorname{EI}_{\mathbb{R}}(2r,r,n)$ exists if and only if:
    \begin{itemize}
        \item[(i)]
        $2(\alpha^2 - 1) < (2\alpha^2-1)n$ and $n \leq \rho_{\mathbb{R}}(r)+1$, or
       
        \smallskip
        
        \item[(ii)]
        $2(\alpha^2-1) = (2 \alpha^2 - 1)n$ and $n \leq \rho_{\mathbb{R}}(r)+2$.
    \end{itemize}
\end{proposition}

Proposition~\ref{prop:d2r} makes it easy to determine the maximum number of subspaces in any $\alpha$-$\operatorname{EI}_{\mathbb{R}}(2r,r,n)$ as a function of $\alpha$ and $r$.
Our main result in this section is that the same characterization applies for $\alpha$-$\operatorname{EI}_{\mathbb{R}}(2r+1,r,n)$, provided $r$ is even and $\alpha \neq \tfrac{1}{2}$.

\begin{theorem}
\label{thm:d2r1}
Let $r\geq 1$ be even, let $n\geq 2$, and let $\alpha \in (0,1)$ satisfy $\alpha \neq \tfrac{1}{2}$.
If $\{W_j\}_{j=1}^n$ is an $\alpha$-$\operatorname{EI}_{\mathbb{R}}(2r+1,r,n)$, then there exists $V \leq \mathbb{R}^{2r+1}$ with $\dim V = 2r$ and $W_j \leq V$ for every~$j$.
Consequently:
    \begin{itemize}
        \item[(a)]
        an $\alpha$-$\operatorname{EI}_{\mathbb{R}}(2r+1,r,n)$ exists if and only if an $\alpha$-$\operatorname{EI}_{\mathbb{R}}(2r,r,n)$ exists (cf.\ Proposition~\ref{prop:d2r}), and
        \smallskip
        
        \item[(b)]
        an $\operatorname{EITFF}_{\mathbb{R}}(2r+1,r,n)$ exists if and only if $r=2$ and $n=5$.
    \end{itemize}
\end{theorem}

\begin{remark}
The hypothesis $\alpha \neq \tfrac{1}{2}$ is essential since a $\tfrac{1}{2}$-$\operatorname{EI}_{\mathbb{R}}(5,2,5)$ exists (Example~\ref{ex: totally symmetric spark bound}) but a $\tfrac{1}{2}$-$\operatorname{EI}_{\mathbb{R}}(4,2,5)$ does not (Proposition~\ref{prop:d2r}).
\end{remark}

Our proof of Theorem~\ref{thm:d2r1} relies on a certain normalization of equi-isoclinic subspaces, detailed in Lemma~\ref{lem:normalization} below.
(We will also use this lemma in Section~\ref{sec: upper bounds}.)
To enunciate it, we first articulate the notion of ``equivalence'' for subspace sequences.

\begin{lemma}[Lemma 2.1 of \cite{FGLI:25}]
\label{lem:equiv}
Suppose that $\{V_j\}_{j=1}^n$ and $\{W_j\}_{j=1}^n$ are sequences of $r$-dimensional subspaces of $\mathbb{F}^d$. 
Let $\{P_j\}_{j=1}^n$ and $\{Q_j\}_{j=1}^n$ be the associated projections, and let $\{\Phi_j\}_{j=1}^n$ and $\{\Psi_j\}_{j=1}^n$ be associated isometries. 
Then the following are equivalent:
\begin{itemize}
\item[(i)] 
there is a unitary $U \in \mathbb{F}^{d \times d}$ such that $W_j = UV_j$ for all $j \in [n]$,
\item[(ii)] 
there is a unitary $U \in \mathbb{F}^{d \times d}$ such that $P_j = U Q_jU^*$ for all $j \in [n]$,
\item[(iii)] 
there exist unitaries $\{Z_j\}_{j=1}^n \in \mathbb{F}^{r \times r}$, $U \in \mathbb{F}^{d \times d}$, such that $\Psi_j = U \Phi_j Z_j$ for all $j \in [n]$,
\item[(iv)] 
there exist unitaries $\{Z_j\}_{j=1}^n \in \mathbb{F}^{r \times r}$ such that $\Psi_i^*\Psi_j = Z_i^*\Phi_i^*\Phi_jZ_j$ for all $i,j \in [n]$. 
\end{itemize}
\end{lemma}
Following \cite{FGLI:25}, we call $\{V_j\}_{j=1}^n$ and $\{W_j\}_{j=1}^n$ \textbf{equivalent} if they satisfy any of the conditions (i)--(iv) of Lemma \ref{lem:equiv}. 
It is easy to see this equivalence preserves the properties of being an $\alpha$-EI, a TFF, or an EITFF.
The following normalization of EIs, with precursors in~\cite{Wong:61,FGLI:25} for the case $d=2r$.

\begin{lemma}
\label{lem:normalization}
Given $d,r \geq 1$ with $d \geq 2r$, and given $n \geq 2$ and $\alpha \in [0,1)$, any $\alpha$-$\operatorname{EI}_{\mathbb{F}}(d,r,n)$ is equivalent to a sequence of subspaces in $\mathbb{F}^d$ with isometries
    \begin{equation}
    \label{eq:normalizedIsometries}
    \Phi_1 = \begin{bmatrix}
        I_r \\
        0 \\
        0
    \end{bmatrix}, 
    \quad
    \Phi_2 = \begin{bmatrix}
        \alpha I_r \\
        \beta I_r \\
        0
    \end{bmatrix}, 
    \quad
    \Phi_{j} = \begin{bmatrix}
        \alpha I_r \\
        X_j \\
        Y_j
    \end{bmatrix} 
    \quad
    \text{for every }j \geq 3,
    \end{equation}
    where $\beta := \sqrt{1 - \alpha^2}$, each $X_j \in \mathbb{F}^{r \times r}$, and each $Y_j \in \mathbb{F}^{(d - 2r) \times r}$. 
  Furthermore, if $\alpha \neq 0$ then for any $\alpha$-$\operatorname{EI}_{\mathbb{F}}(d,r,n)$ with isometries as in~\eqref{eq:normalizedIsometries} and every $j \geq 3$, the following hold:
 
 \begin{itemize}
 \item[(a)] $U_j:=\alpha I_r + \frac{\beta}{\alpha} X_j$ is unitary,
 \item[(b)] $X_j$ is an invertible normal matrix, and
 \item[(c)] $Y_j^*Y_j = \tfrac{1 - 3 \alpha^2}{1 - \alpha^2}I_r + \frac{\alpha^3}{1-\alpha^2}(U_j + U_j^*)$.  
 \end{itemize}
    
    \end{lemma}

To clarify, if $d = 2r$ above then the isometries take the form
\[
    \Phi_1 = \begin{bmatrix}
        I_r \\
        0
    \end{bmatrix}, 
    \quad
    \Phi_2 = \begin{bmatrix}
        \alpha I_r \\
        \beta I_r
    \end{bmatrix}, 
    \quad
    \Phi_{j} = \begin{bmatrix}
        \alpha I_r \\
        X_j
    \end{bmatrix} 
    \quad
    \text{for every }j \geq 3
\]
and in (c) we interpret $Y_j \in \mathbb{F}^{0 \times r}$ to satisfy $Y_j^* Y_j = 0 \in \mathbb{F}^{r \times r}$ for every $j \geq 3$.

\begin{proof}[Proof of Lemma~\ref{lem:normalization}]
We may assume $\alpha \neq 0$, or else the lemma is trivial.
Let $\{\Psi_j\}_{j=1}^n$ be the isometries of an $\alpha$-$\operatorname{EI}_\mathbb{F}(d,r,n)$.
First, select any unitary $V \in \mathbb{F}^{d\times d}$ such that $V \Psi_1 = \Phi_1$.
To match the top blocks of the remaining isometries with~\eqref{eq:normalizedIsometries}, write
\[
V\Psi_j = \begin{bmatrix}
    A_j \\
    B_j \\
    C_j\\
\end{bmatrix}
\quad
\text{for every }j \geq 2
\]
with $A_j \in \mathbb{F}^{r \times r}$, $B_j \in \mathbb{F}^{r \times r}$, and $C_j \in \mathbb{F}^{(d-2r) \times r}$.
Here,
\[
A_j = \Phi_1^*(V \Psi_j)
= (V \Psi_1)^*(V \Psi_j) 
= \Psi_1^* \Psi_j,
\]
and $Z_j := \alpha^{-1} A_j^*$ is unitary by Lemma~\ref{lem: isoclinism}(iii).
Then
\[
V\Psi_jZ_j = 
\begin{bmatrix}
    \alpha I_r \\
    \alpha^{-1}B_jA_j^* \\
    \alpha^{-1}C_jA_j^* \\
\end{bmatrix}
\quad
\text{for every }j \geq2.
\]
Next, we make the second isometry match~\eqref{eq:normalizedIsometries}.
Abbreviate
\[
V \Psi_2 Z_2
=:
\begin{bmatrix}
    \alpha I_r \\
    B \\
    C
\end{bmatrix},
\]
that is, $B = \alpha^{-1} B_2 A_2^*$ and $C = \alpha^{-1} C_2 A_2^*$.
Since $V \Psi_2 Z_2$ is an isometry,
\[
\begin{bmatrix} B^* & C^* \end{bmatrix} \begin{bmatrix} B \\ C \end{bmatrix}
= (1-\alpha^2)I_r 
= \begin{bmatrix} \beta I_r & 0 \end{bmatrix} \begin{bmatrix} \beta I_r \\ 0 \end{bmatrix},
\]
so there exists a unitary $W \in \mathbb{F}^{(d-r) \times (d-r)}$ with 
\[
W \begin{bmatrix} B \\ C \end{bmatrix} = \begin{bmatrix} \beta I_r \\ 0 \end{bmatrix}.
\]
Now define $U := \left[\begin{smallmatrix}
    I_r & 0 \\
    0 & W
\end{smallmatrix}\right] V$ and $Z_1:=I_r$.
Then
\[
U \Psi_1 Z_1
= \begin{bmatrix}
    I_r & 0 \\
    0 & W
\end{bmatrix}V \Psi_1 Z_1
= \begin{bmatrix}
    I_r & 0 \\
    0 & W
\end{bmatrix} \begin{bmatrix} I_r \\ 0 \end{bmatrix} = \Phi_1
\]
and
\[
U \Psi_2 Z_2 = 
\begin{bmatrix}
    I_r & 0 \\
    0 & W
\end{bmatrix}V \Psi_2 Z_2 = \begin{bmatrix}
    \alpha I_r \\
    \beta I_r \\
    0
\end{bmatrix} = \Phi_2,
\]
while for every $j \geq 3$,
\[
U\Psi_jZ_j = \begin{bmatrix}
    I_r & 0 \\
    0 & W
\end{bmatrix}V \Psi_j Z_j
=
\begin{bmatrix}
\alpha I_r \\
X_j \\
Y_j
\end{bmatrix}
= \Phi_j
\]
for some $X_j \in \mathbb{F}^{r \times r}$ and $Y_j \in \mathbb{F}^{(d-2r) \times r}$, as desired.

For the ``furthermore'' part, select any $j \geq 3$.
By Lemma~\ref{lem: isoclinism}(iii), $\Phi_2^*\Phi_j = \alpha^2 I_r + \beta X_j$ satisfies 
\[
I_r = \tfrac{1}{\alpha^2}(\Phi_2^*\Phi_j)(\Phi_2^*\Phi_j)^*
= (\alpha I_r + \tfrac{\beta}{\alpha}X_j)(\alpha I_r + \tfrac{\beta}{\alpha}X_j)^*.
\]
Thus, $U_j:=\alpha I_r + \frac{\beta}{\alpha}X_j$ is unitary, and $X_j = \frac{\alpha}{\beta}(U_j - \alpha I_r)$ is normal.
Furthermore, $X_j$ is invertible since $\tfrac{\beta}{\alpha}X_j = U_j - \alpha I_r$ and $\alpha \in (0,1)$ is not an eigenvalue of the unitary $U_j$.
Finally, since $\Phi_j$ is an isometry, 
\[
Y_j^*Y_j = (1 - \alpha^2)I_r - X_j^*X_j
= (1 - \alpha^2)I_r - \frac{\alpha^2}{\beta^2}(U_j - \alpha I_r)^*(U_j-\alpha I_r),
\]
which simplifies to~(c).
\end{proof}

Now we prove this section's main result.

\begin{proof}[Proof of Theorem~\ref{thm:d2r1}]
We may assume $n \geq 3$.
Let $\{\Phi_j\}_{j=1}^n$ be corresponding isometries for $\{W_j\}_{j=1}^n$.
By passing to an equivalent subspace sequence, we assume without loss of generality that the isometries are normalized as in~\eqref{eq:normalizedIsometries}.
In that notation, we show $\operatorname{rank} Y_j^* Y_j = 0$ for each $j \geq 3$, and it follows that each $W_j$ lies in the span of the first $2r$ standard basis vectors.

Take any $j \geq 3$ and consider $Y_j \in \mathbb{R}^{1 \times r}$.
Since $Y_j^* Y_j \in \mathbb{R}^{r \times r}$ has rank at most~1, and since $r$ is even, it suffices to show $0$ has even multiplicity as an eigenvalue of $Y_j^* Y_j$.
Recall from Lemma~\ref{lem:normalization}(c) that $Y_j^*Y_j = \tfrac{1 - 3 \alpha^2}{1 - \alpha^2}I_r + \frac{\alpha^3}{1 - \alpha^2}(U_j + U_j^*)$ for some unitary $U_j \in \mathbb{R}^{r\times r}$.
By an application of the spectral theorem, the multiplicity of $0$ as an eigenvalue of $Y_j^*Y_j$ is the sum of the multiplicities of the eigenvalues of $U_j$ with real part $\frac{1 - 3 \alpha^2}{2 \alpha^3}$.
Since $U_j$ is a real unitary, it therefore suffices to show any such $\lambda$ has nonzero imaginary part, that is, $\frac{1 - 3 \alpha^2}{2 \alpha^3} \notin \{ \pm 1 \}$.
Indeed, the equation $\frac{1 - 3 \alpha^2}{2 \alpha^3} = -1$ can be rewritten as 
\[
0
= -2\alpha^3 + 3\alpha^2 - 1 
= -(\alpha-1)^2(2\alpha+1),
\]
which has no solutions since $\alpha \in (0,1)$.
Likewise, the equation $\frac{1 - 3 \alpha^2}{2 \alpha^3} = 1$ can be rewritten as
\[
0
= 2\alpha^3 + 3\alpha^2-1 
= (\alpha + 1)^2(2 \alpha - 1),
\]
which has no solutions since $\alpha \neq \tfrac{1}{2}$ by assumption.
This completes the proof that a $2r$-dimensional subspace contains each $W_j$.

Part~(a) follows immediately.
For the forward direction of (b), suppose an $\operatorname{EITFF}_{\mathbb{R}}(2r+1,r,n)$ exists.
Then its parameter of isoclinism is $\alpha = \tfrac{1}{2}$ since the subspaces of a fusion frame for $\mathbb{R}^{2r+1}$ cannot all reside in a common space of dimension $2r$.
Comparing with~\eqref{eq:EITFF alpha}, we find
\[ 
\frac{1}{2} 
= \sqrt{\frac{nr-(2r+1)}{(2r+1)(n-1)}},
\]
so that
\[
n
= \frac{6r+3}{2r-1}.
\]
Since $n$ is an integer, $2r-1$ divides $6r+3$. 
The only even $r$ for which this holds is $r=2$, since when $r \geq 4$ we find 
\[
3(2r-1) 
< 6r+3 
< 4(2r-1).
\] 
Thus, $r=2$ and $n=5$.
Finally, the reverse direction of (b) holds since an $\operatorname{EITFF}_{\mathbb{R}}(5,2,5)$ exists (Example~\ref{ex: totally symmetric spark bound}).
\end{proof}

\begin{remark}
The proof of Theorem~\ref{thm:d2r1} depends heavily on the base field being~$\mathbb{R}$.
We do not know if a similar result holds over the complex numbers.
\end{remark}


\section{Upper bounds for equi-isoclinic subspaces}
\label{sec: upper bounds}

Next, we prove a new absolute upper bound on $v_{\mathbb{F}}(r,d)$.
Recall that $\mathbb{F}_{H}^{d\times d} = \{ M \in \mathbb{F}^{d\times d} : M^* = M \}$ denotes the real vector space of Hermitian $d\times d$ matrices with entries in $\mathbb{F}$, with
\[
\operatorname{dim} \mathbb{F}_{H}^{d\times d} = 
\begin{cases}
\tfrac{d(d+1)}{2} & \text{if }\mathbb{F} = \mathbb{R}, \\[3 pt]
d^2 & \text{if } \mathbb{F} = \mathbb{C}.
\end{cases}
\]
We prove the following.

\begin{theorem}
\label{thm:biggerzonbound}
Choose any $\alpha \in [0,1)$ and $n \geq 3$.
If an $\alpha$-$\operatorname{EI}_{\mathbb{F}}(d,r,n)$ exists, then
\[
n \leq \dim \mathbb{F}_{H}^{d\times d} - 3 \dim \mathbb{F}_{H}^{r\times r} + 3.
\]
\end{theorem}

This improves on the bound $n \leq \dim \mathbb{F}_{H}^{d\times d} - \dim \mathbb{F}_{H}^{r\times r} + 1$ of Lemmens and Seidel whenever $r>1$~\cite{LemmensS:73,Hoggar:76,FGLI:25}.
Our argument refines their technique, which begins with the following observation of Gerzon~\cite{LS:73}.

\begin{proposition}
\label{prop: independence}
Choose any $\alpha \in [0,1)$, and let $\{ P_j \}_{j=1}^n$ be the orthogonal projections of an $\alpha$-$\operatorname{EI}_{\mathbb{F}}(d,r,n)$.
Then $\{P_j\}_{j=1}^n$ is a linearly independent sequence in the real vector space $\mathbb{F}_H^{d \times d}$.
\end{proposition}

\begin{proof}
For $i \neq j,$ 
\[
\langle P_j, P_i \rangle_F = \operatorname{tr}(P_i^*P_j) 
= \operatorname{tr}(P_iP_iP_j) 
= \operatorname{tr}(P_iP_jP_i) 
= \operatorname{tr}(\alpha^2P_i) = \alpha^2 r,
\] 
since $P_i P_j P_i = \alpha^2 P_i$ by isoclinism.
Furthermore, $\langle P_j, P_j \rangle_F = \operatorname{rank}P_j = r$, so the Gram matrix of $\{P_j\}_{j=1}^n$ is $\alpha^2rJ_n + (r-\alpha^2r)I_n$, which has full rank since $\alpha \neq 1$.
\end{proof}

In Proposition~\ref{prop: independence}, linear independence gives the naive bound $n \leq \dim \mathbb{F}_H^{d\times d}$.
However, this can be improved by replacing $\mathbb{F}_H^{d\times d}$ with a smaller subspace that contains each of $P_1,\ldots,P_n$.
For instance, Lemmens and Seidel observe that every $P_j$ belongs to the vector space $\{ M \in \mathbb{F}_H^{d \times d} : P_1 M P_1 = c P_1 \text{ for some }c \in \mathbb{R} \}$ since $P_1^3 = P_1$ and $P_1P_jP_1 = \alpha^2 P_1$ for every $j >1$.
To further improve their bound, we intersect many such subspaces, as follows.

\begin{definition}
Given an $\alpha$-$\operatorname{EI}_\mathbb{F}(d,r,n)$ with orthogonal projections $\{ P_j \}_{j=1}^n$, and given $j \in [n]$, the $j$th \textbf{corner matrix space} is the real vector space
\[ 
\mathcal{K}_j:= \bigcap_{i=1}^j \big\{M \in \mathbb{F}_H^{d \times d}: P_i M P_i \in \operatorname{span}\{P_i\} \big\} \leq \mathbb{F}_H^{d \times d}.
\]
Equivalently, for any choice of isometries $\{ \Phi_j \}_{j=1}^n$ with $\Phi_j \Phi_j^* = P_j$,
\begin{equation}
\label{eq: Kj isoms}
\mathcal{K}_j = \bigcap_{i=1}^j \big\{M \in \mathbb{F}_H^{d \times d}: \Phi_i^*M\Phi_i \in \operatorname{span} \{I_r\} \big\}.
\end{equation}
\end{definition}
Then every $P_i \in \mathcal{K}_j$, and Proposition~\ref{prop: independence} gives the following.

\begin{proposition}
\label{prop: korner bound}
Choose any $\alpha \in [0,1)$.
Then for any $\alpha$-$\operatorname{EI}_{\mathbb{F}}(d,r,n)$ with corner matrix spaces $\mathcal{K}_j \leq \mathbb{F}^{d \times d}$, $j \in [n]$, it holds that
\[
n 
\leq \dim \mathcal{K}_j 
\quad
\text{for every }
j \in [n].
\]
\end{proposition}

Larger values of $j$ give tighter bounds since $\mathcal{K}_n \leq \cdots \leq \mathcal{K}_1$, but it becomes harder to control $\dim \mathcal{K}_j$ as $j$ grows.
To prove Theorem~\ref{thm:biggerzonbound}, we compute $\dim \mathcal{K}_j$ for $j=1,2,3$.

\begin{remark}
\label{rem: KJ}
There are two possible variants of our approach that we did not follow.
First, our definition of $\mathcal{K}_j$ depends on the order of the subspaces in the sequence $\{W_j\}_{j=1}^n$.
More generally, for each $J \subseteq [n]$, one might define
\[
\mathcal{K}_J
:= \bigcap_{i\in J} \big\{M \in \mathbb{F}_H^{d \times d}: P_i M P_i \in \operatorname{span}\{P_i\} \big\} 
\leq \mathbb{F}_H^{d \times d},
\]
and then of course it holds that $n \leq \dim \mathcal{K}_J$.
We found no need for this level of generality.
In fact, our results imply $\dim \mathcal{K}_J$ depends only on $|J|$ when $|J| \leq 3$.

Second, considering the order of quantifiers in the expression
\[
\mathcal{K}_J 
= \{ M \in \mathbb{F}_H^{d\times d} : \forall i \in J, \, \exists c_i \in \mathbb{R} \text{ s.t.\ } P_i M P_i = c_i P_i \} \leq \mathbb{F}_H^{d \times d},
\]
one might be inspired to define
\[
\mathcal{L}_J 
:= \{ M \in \mathbb{F}_H^{d\times d} : \exists c\in \mathbb{R} \text{ s.t.\ } \forall i \in J, \, P_i M P_i = c P_i \} \leq \mathbb{F}_H^{d \times d}.
\]
Here, $P_i \in \mathcal{L}_J$ whenever $i \notin J$, which gives the bound $n \leq \dim \mathcal{L}_J + |J|$.
However, a better bound holds since one can show that 
\[
\dim \mathcal{L}_J 
= \dim \mathcal{K}_J -|J|+1,
\]
so that $n \leq \dim \mathcal{L}_J + |J| - 1$.
(Indeed, $\mathcal{L}_J$ is the preimage of the span of the all-ones vector under the linear map $\mathcal{K}_J \to \mathbb{R}^J \colon M \mapsto \begin{bmatrix} \langle M, P_j \rangle_F \end{bmatrix}_j$.)
In practice, we found it easier to control $\dim \mathcal{K}_J$ than $\dim \mathcal{L}_J$.
\end{remark}

We now compute the dimensions of $\mathcal{K}_1$, $\mathcal{K}_2$, and $\mathcal{K}_3$.
Our proofs rely on the normalization of Lemma~\ref{lem:normalization}.
In that context, note that $\dim \mathcal{K}_j$ is preserved under equivalence of subspace sequences, as seen by conjugating with an appropriate unitary.
This leads to the following result, versions of which appear in~\cite{LemmensS:73,Hoggar:76}.

\begin{proposition}
\label{prop:dimK1}
Suppose $\alpha \in [0,1)$ and $n \geq 2$. Then for any \fei{d}{r}{n} with corner matrix spaces $\mathcal{K}_j \leq \mathbb{F}_H^{d \times d}$, $j \in [n]$, it holds that 
\[
\dim\mathcal{K}_1 
= \dim \mathbb{F}_H^{d \times d} - \dim \mathbb{F}_H^{r \times r} + 1.
\]
Specifically, for any $\alpha$-$\operatorname{EI}_{\mathbb{F}}(d,r,n)$ with isometries as in~\eqref{eq:normalizedIsometries}, $\mathcal{K}_1$ consists of all
\[
M 
= \begin{bmatrix}
    cI_r & A & C \\
    A^* & B & D \\
    C^* & D^* & E
\end{bmatrix}
\in \mathbb{F}_H^{d \times d}
\]
with $A \in \mathbb{F}^{r \times r}$, $B \in \mathbb{F}^{r \times r}_H$, $C,D \in \mathbb{F}^{r \times (d-2r)}$, $E \in \mathbb{F}^{(d-2r) \times (2-dr)}_H$, and $c \in \mathbb{R}$.
\end{proposition} 

Next, we leverage Proposition~\ref{prop:dimK1} to compute $\dim \mathcal{K}_2$.

\begin{theorem}
\label{thm:dimK2}
Suppose $\alpha \in [0,1)$ and $n \geq 2$. 
Then for any \fei{d}{r}{n} with corner matrix spaces $\mathcal{K}_j \leq \mathbb{F}_H^{d \times d}$, $j \in [n]$, it holds that 
\[
\dim\mathcal{K}_2 
= \dim \mathbb{F}_H^{d \times d} - 2\dim \mathbb{F}_H^{r \times r} + 2.
\]
Specifically, for any $\alpha$-$\operatorname{EI}_{\mathbb{F}}(d,r,n)$ with isometries as in~\eqref{eq:normalizedIsometries}, $\mathcal{K}_2$ consists of all
\[
M=
\left[ \begin{array}{ccc}
c I_r & A & C \\
A^* & B & D \\
C^* & D^* & E
\end{array} \right]
\in \mathbb{F}^{d \times d}_H
\]
with $A \in \mathbb{F}^{r \times r}$, $B \in \mathbb{F}^{r \times r}_H$, $C,D \in \mathbb{F}^{r \times (d-2r)}$, $E \in \mathbb{F}^{(d-2r) \times (2-dr)}_H$, and $c \in \mathbb{R}$, such that
\[
\alpha A + \alpha A^* + \beta B \in \operatorname{span}\{I_r\},
\]
where $\beta = \sqrt{1-\alpha^2}$.
\end{theorem}

\begin{proof} 
Let $\{\Phi_j\}_{j=1}^n$ be the isometries of an \fei{d}{r}{n} with $\alpha \in [0,1)$ and $n \geq 2$, normalized as in~\eqref{eq:normalizedIsometries}.
Then $\mathcal{K}_2$ consists of all $M \in \mathcal{K}_1$ with $\Phi_2^* M \Phi_2 \in \operatorname{span}\{I_r\}$.
With this in mind, choose any $M \in \mathcal{K}_1$.
By Proposition~\ref{prop:dimK1}, 
\begin{equation}
\label{eq: block M}
M 
= \begin{bmatrix}
        c I_r & A & C \\
        A^* & B & D \\
        C^* & D^* & E
\end{bmatrix},
\end{equation}
for some $A \in \mathbb{F}^{r \times r}$, $B \in \mathbb{F}^{r \times r}_H$, $C,D \in \mathbb{F}^{r \times (d-2r)}$, $E \in \mathbb{F}^{(d-2r) \times (2-dr)}_H$, and $c \in \mathbb{R}$. 
Now we calculate 
\[
\Phi_2^*M\Phi_2 
= \begin{bmatrix}
    \alpha I_r & \beta I_r & 0
\end{bmatrix}\begin{bmatrix}
        c I_r & A & C \\
        A^* & B & D \\
        C^* & D^* & E
\end{bmatrix} \begin{bmatrix}
    \alpha I_r \\
    \beta I_r \\
    0
\end{bmatrix}
= \alpha^2c I_r + \alpha \beta A + \alpha \beta A^* + \beta^2 B
\]
to see that $M \in \mathcal{K}_2$ if and only if $\alpha A + \alpha A^* + \beta B \in \operatorname{span}\{I_r\}$.
To find the dimension of $\mathcal{K}_2$, consider the linear map $T \colon \mathcal{K}_1 \to \mathbb{F}_H^{r \times r}$ given by
\[
T(M) 
= \alpha A + \alpha A^* + \beta B,
\]
for any $M \in \mathcal{K}_1$ partitioned as in~\eqref{eq: block M}.
Then $\mathcal{K}_2$ is the preimage of $\operatorname{span}\{ I_r \}$ under $T$, so that
\[
\dim \mathcal{K}_2 
= \dim \ker T + 1 
= \dim \mathcal{K}_1 - \operatorname{rank} T + 1.
\]
Furthermore, $T$ surjects $\mathcal{K}_1$ onto $\mathbb{F}_H^{r \times r}$. 
(Given $S \in \mathbb{F}_H^{r \times r}$, take any $M \in \mathcal{K}_1$ with $A=0$ and $B = \frac{1}{\beta}S$ in~\eqref{eq: block M} to find $T(M) = S$.)
Applying Proposition~\ref{prop:dimK1}, we find
\[
\dim \mathcal{K}_2
= \dim \mathcal{K}_1 - \operatorname{rank} T + 1
= \dim \mathbb{F}_H^{d \times d} - 2\dim \mathbb{F}_H^{r \times r} + 2.
\]
The theorem follows since any \fei{d}{r}{n} is equivalent to one with isometries as in~\eqref{eq:normalizedIsometries}.
\end{proof}

Finally, we compute $\dim \mathcal{K}_3$.
The following proves Theorem~\ref{thm:biggerzonbound}.

\begin{theorem}
\label{thm: dim K3}
Suppose $\alpha \in [0,1)$ and $n \geq 3$.
Then for any $\alpha$-$\operatorname{EI}_{\mathbb{F}}(d,r,n)$,
\[
\dim \mathcal{K}_3
= \dim\mathbb{F}_H^{d \times d} - 3 \dim \mathbb{F}_H^{r \times r} + 3.
\]
\end{theorem}

\begin{proof}
Let $\{\Phi_j\}_{j=1}^n$ be the isometries of an \fei{d}{r}{n} with $\alpha \in [0,1)$ and $n \geq 3$.
Without loss of generality, we assume these isometries are normalized as in~\eqref{eq:normalizedIsometries}.
By~\eqref{eq: Kj isoms}, $\mathcal{K}_3$ consists of all $M \in \mathcal{K}_2$ with $\Phi_3^* M \Phi_3 \in \operatorname{span}\{I_r\}$.
With this in mind, consider the linear map $T \colon \mathcal{K}_2 \to \mathbb{F}^{r \times r}_H$ given by $T(M) = \Phi_3^*M\Phi_3$.
Then $\mathcal{K}_3$ is the preimage of $\operatorname{span}\{I_r\}$ under $T$, so that
\[
\dim \mathcal{K}_3 
= \dim \ker T + 1 
= \dim \mathcal{K}_2 - \operatorname{rank}T + 1.
\]
We will show that $T$ maps $\mathcal{K}_2$ onto $\mathbb{F}_H^{r \times r}$, and it will follow from Theorem~\ref{thm:dimK2} that
\[
\dim \mathcal{K}_3 
= \dim \mathcal{K}_2 - \operatorname{rank}T + 1 
= \dim\mathbb{F}_H^{d \times d} - 3 \dim \mathbb{F}_H^{r \times r} + 3.
\]

We argue in cases.
First suppose $\alpha = 0$.
Then the hypothesis $n \geq 3$ implies $d \geq 3r$.
Furthermore, referring to~\eqref{eq:normalizedIsometries}, we have $X_3 = \Phi_2^* \Phi_3 = 0$, from which it follows that $Y_3 \in \mathbb{F}^{(d-2r) \times r}$ is an isometry.
Given a desired output $S \in \mathbb{F}_H^{r \times r}$, define
\[
M := \begin{bmatrix}
0 & 0 & 0 \\
0 & 0 & 0 \\
0 & 0 & E
\end{bmatrix}
\in \mathbb{F}_H^{d \times d},
\qquad
E := Y_3 S Y_3^* \in \mathbb{F}_H^{(d-2r) \times (d-2r)}.
\]
Then $M \in \mathcal{K}_2$ by Theorem~\ref{thm:dimK2}, and it is easy to check that $T(M) = S$.

Now suppose $\alpha \neq 0$.
Referring to~\eqref{eq:normalizedIsometries}, we claim that $I_r - \tfrac{1}{\beta} X_3^*$ is invertible, that is, $\beta$ is not an eigenvalue of $X_3^*$.
By Lemma \ref{lem:normalization}(a), $X_3^* = [\tfrac{\alpha}{\beta}(U_3 - \alpha I)]^*$ for some unitary $U_3$, so every eigenvalue of $X_3^*$ takes the form $\tfrac{\alpha}{\beta}(\lambda - \alpha)$, where $\lambda$ is a unimodular complex number.
Furthermore, if $\beta = \tfrac{\alpha}{\beta}(\lambda-\alpha)$ then $\lambda$ is real, hence $\lambda \in \{\pm 1\}$.
However, when $\lambda = -1$, $\tfrac{\alpha}{\beta}(\lambda - \alpha) 
< 0 
< \beta,$
and when $\lambda =1$, dividing both sides of the inequality $\alpha - \alpha^2 < 1 - \alpha^2$ by $\beta > 0$ gives $\tfrac{\alpha}{\beta}(\lambda-\alpha) < \beta$.
This proves the claim.

Next, given a desired output $S \in \mathbb{F}_H^{r \times r}$ of $T$, define 
\[
A:= \tfrac{1}{2\alpha}( I_r - \tfrac{1}{\beta} X_3^*)^{-1} S X_3^{-1} \in \mathbb{F}^{r \times r},
\]
where $I - \tfrac{1}{\beta} X_3^*$ is invertible by the claim above and $X_3$ is invertible by Lemma~\ref{lem:normalization}(b).
Then define $B:=-\tfrac{\alpha}{\beta}(A+A^*) \in \mathbb{F}_H^{r \times r}$ and
\[
M:= \begin{bmatrix}
0 & A & 0 \\
A^* & B & 0 \\
0 & 0 & 0
\end{bmatrix}
\in \mathbb{F}_H^{d \times d}.
\]
Then $M \in \mathcal{K}_2$ by Theorem~\ref{thm:dimK2}, and we have
\begin{align*}
T(M) 
= \Phi_3^* M \Phi_3
&= \begin{bmatrix} \alpha I_r & X_3^* & Y_3^* \end{bmatrix}
\begin{bmatrix}
0 & A & 0 \\
A^* & B & 0 \\
0 & 0 & 0
\end{bmatrix}
\begin{bmatrix} \alpha I_r \\ X_3 \\ Y_3 \end{bmatrix} \\[5 pt]
&= \alpha X_3^* A^* + \alpha A X_3 + X_3^* B X_3 \\[5 pt]
&= \alpha X_3^* A^* (I_r - \tfrac{1}{\beta} X_3) + \alpha(I_r - \tfrac{1}{\beta} X_3^*) A X_3.
\end{align*}
Here, $\alpha(I_r - \tfrac{1}{\beta} X_3^*) A X_3 = \tfrac{1}{2} S$ by the definition of $A$, so $T(M) = S$, as desired.
\end{proof}

\begin{example}
\label{ex:K3 tight}
Equality is achieved in the bound $n \leq \dim \mathcal{K}_3 = \mathbb{F}_H^{d \times d} - 3 \mathbb{F}_H^{r \times r} + 3$ for an $\operatorname{EITFF}_{\mathbb{R}}(4,2,4)$, that is, an $\tfrac{1}{\sqrt{3}}$-$\operatorname{EI}_{\mathbb{R}}(4,2,4)$, whose existence follows from Proposition~\ref{prop:d2r}.
(In fact, one can produce such an EITFF by applying Hoggar's ``$\mathbb{C}$-to-$\mathbb{R}$'' trick~\cite{Hoggar:77} to an $\operatorname{EITFF}_{\mathbb{C}}(2,1,4)$, which can in turn be obtained from a simplex in the Bloch sphere $\mathbb{C} P^1$.)
\end{example}

\begin{example}
For an $\alpha$-$\operatorname{EI}_{\mathbb{F}}(d,r,n)$ with $\alpha \in [0,1)$, one might be inclined to define $\mathcal{K}_0 := \mathbb{F}^{d \times d}_H$, since it then holds that $\mathcal{K}_n \leq \cdots \leq \mathcal{K}_1 \leq \mathcal{K}_0$ and $n \leq \dim \mathcal{K}_j$ for every $j \geq 0$.
With this definition, the results of Proposition~\ref{prop:dimK1}, Theorem~\ref{thm:dimK2}, and Theorem~\ref{thm: dim K3} are neatly summarized as follows:
\[
\dim \mathcal{K}_j = \dim \mathbb{F}^{d \times d}_H - j \dim \mathbb{F}^{r \times r}_H + j
\quad
\text{for }
j \in \{ 0,1,2,3 \}.
\]
However, this pattern does not continue, and there are many examples where
\[
\dim \mathcal{K}_4 
\neq \dim \mathbb{F}_H^{d \times d} - 4 \mathbb{F}_H^{r \times r} + 4.
\]
For instance, the $\operatorname{EITFF}_{\mathbb{R}}(4,2,4)$ from Example~\ref{ex:K3 tight} has $\dim \mathcal{K}_4 = 4$ since
\[
4 = n \leq \dim \mathcal{K}_4 \leq \dim \mathcal{K}_3 = 4, 
\]
while $\dim \mathbb{R}_H^{4 \times 4} - 4 \mathbb{R}_H^{2 \times 2} + 4 = 2$.
\end{example}

\begin{example}
In contrast with Theorem~\ref{thm: dim K3}, the dimension of $\mathcal{K}_4$ cannot be expressed as a function of $d$, $r$, and $\mathbb{F}$, as demonstrated by numerical calculations for the following examples with $(d,r,\mathbb{F})=(5,2,\mathbb{R})$.
We find $\dim \mathcal{K}_4 = 9$ for the $\tfrac{1}{\sqrt{3}}$-$\operatorname{EI}_{\mathbb{R}}(5,2,4)$ obtained by injecting the subspaces of $\mathbb{R}^4$ from Example~\ref{ex:K3 tight} into $\mathbb{R}^5$.
Meanwhile, we find $\dim \mathcal{K}_4 = 7$ for the $\operatorname{EITFF}_{\mathbb{R}}(5,2,5)$ in Example~3.4 of~\cite{FIJM:24}.
\end{example}

\begin{example}
Unlike $\mathcal{K}_1$, $\mathcal{K}_2$, and $\mathcal{K}_3$, the dimension of $\mathcal{K}_4$ may depend on the ordering of the subspaces.
That is, $\dim \mathcal{K}_J$ may depend on $J$ and not only on $|J|$.
In particular, $\dim \mathcal{K}_4$ cannot be expressed as a function of $(\alpha,d,r,\mathbb{F})$.
We observe this numerically for the $\operatorname{EITFF}_{\mathbb{C}}(10,4,10)$ given in Example~3.5 of~\cite{FIJM:24}:
working with the MATLAB file from~\cite{FIJM:24}, we find $\dim \mathcal{K}_{\{1,2,3,4\}}=40$ but $\dim \mathcal{K}_{\{1,2,4,8\}}=42$.
Interestingly, the collection of all $J$ with $|J|=4$ and $\dim \mathcal{K}_J = 42$ forms a $(10,4,2)$ design, that is, a $\operatorname{BIBD}(10,15,6,4,2)$.
(This is a consequence of symmetry, where this EITFF has doubly homogeneous automorphism group.)
\end{example}


\section{Achieving $\dim \mathcal{K}_n = n$}
\label{sec: dim Kn}

For $j > 3$, it seems to be more challenging to bound $\dim \mathcal{K}_j$, and further progress may require additional techniques.
As a preliminary test of the value of such bounds, we now investigate saturation of the final inequality $n \leq \dim \mathcal{K}_n$.
In this section, we show this occurs infinitely often in both the real and complex settings.
Such examples are necessarily EITFFs, by the following proposition, whose proof follows that for equiangular lines attaining Gerzon's bound~\cite{LS:73}.

\begin{proposition}
If $\alpha \neq 1$, then any $\alpha$-$\operatorname{EI}_{\mathbb{F}}(d,r,n)$ with $\dim \mathcal{K}_n = n$ is an $\operatorname{EITFF}_{\mathbb{F}}(d,r,n)$.
\end{proposition}

\begin{proof}
Let $P_1,\ldots,P_n \in \mathbb{F}^{d \times d}$ be the orthogonal projections of such an $\operatorname{EI}$.
By Proposition~\ref{prop: independence} and its proof, these projections form a basis for $\mathcal{K}_n$, and $\langle P_i, P_j \rangle_F = \alpha^2 r$ when $i \neq j$.
Since $I \in \mathcal{K}_n$, there exist $c_1,\ldots,c_n \in \mathbb{R}$ to expand $I = \sum_{j=1}^n c_j P_j$, and it suffices to show $c_1 = \cdots = c_n$.
Taking a trace, we find $\sum_{j=1}^n c_j = d/r$.
Calculating $\langle I, P_i \rangle_F$ in two different ways, we find 
\[
r = c_i r + \sum_{j \neq i} c_j \alpha^2 r 
= c_i r + \alpha^2 r(d/r) - c_i \alpha^2 r.
\]
Thus,
\[
c_i = \frac{ 1 - \alpha^2(d/r)}{1-\alpha^2}
\qquad
\text{for every }i.
\qedhere
\]
\end{proof}

\begin{example}
In the case where $r=1$, it follows from~\eqref{eq: Kj isoms} that any $\alpha$-$\operatorname{EI}_{\mathbb{F}}(d,1,n)$ has $\mathcal{K}_j = \mathbb{F}_H^{d \times d}$ for every $j \leq n$.
In particular, when $r=1$ the inequality $n \leq \dim \mathcal{K}_n$ reduces to Gerzon's bound $n \leq \operatorname{dim} \mathbb{F}_H^{d\times d}$.
When $\mathbb{F} = \mathbb{R}$, examples that saturate this bound are known for $d \in \{2,3,7,23\}$.
It is conjectured that no other examples exist~\cite[p263]{GodsilR:01}.
By contrast, when $\mathbb{F} = \mathbb{C}$, Gerzon's bound is known to be saturated in many dimensions, and \textit{Zauner's conjecture} asserts that an example exists for every $d$~\cite{Zauner:11,RenesBSC:04}.
This is a major open problem~\cite{HorodeckiRZ:22}.
\end{example}

\begin{example}
Examples~3.5--3.8 of~\cite{FIJM:24} construct, respectively, an $\operatorname{EITFF}_{\mathbb{C}}(10,4,10)$, an $\operatorname{EITFF}_{\mathbb{R}}(11,3,11)$, an $\operatorname{EITFF}_{\mathbb{C}}(16,5,12)$, and an $\operatorname{EITFF}_{\mathbb{C}}(10,3,15)$.
Working numerically with the MATLAB files from~\cite{FIJM:24}, we find $\operatorname{dim} \mathcal{K}_n = n$ in each case.
\end{example}

\begin{example}
\label{ex: lines}
Given any integer $b\geq 2$, let $C_b = \tfrac{1}{b+1}\binom{2b}{b}$ be the $b$th Catalan number.
Selecting $a=2$ in Theorem~5.3(i) of~\cite{FIJM:24} constructs an $\operatorname{EITFF}_{\mathbb{R}}(C_{b+1},C_b,2b+1)$ with the surprising property that any permutation of subspaces can be achieved by an appropriate orthogonal transformation.
For each $b\in \{2,3,4,5\}$, we observe numerically that this EITFF satisfies $\dim \mathcal{K}_n = n$.
The authors do not know if this holds for every $b \geq 2$
(and would love to find out).
\end{example}

In the remainder of this section, we focus on EIs consisting of subspaces with dimension half that of the ambient space.
As mentioned in the introduction, the existence of such EIs has been completely resolved in terms of the \textbf{Radon--Hurwitz number}, given as follows for $r=(2a+1)2^{4b+c}$ with $0 \leq c \leq 3$:
\begin{equation}
\label{eq: radon hurwitz number values}
    \rho_\mathbb{F}(r) = \begin{cases}
        8b+2^c & \text{if }\mathbb{F} = \mathbb{R}, \\
        8b+2c+2 & \text{if }\mathbb{F} = \mathbb{C}. \\
    \end{cases}
\end{equation}

\begin{proposition}[\cite{LemmensS:73,Hoggar:76,FGLI:25}]
\label{prop: EITFF d=2r existence}
For any $r \geq 1$,
    \begin{itemize}
    \item[(a)]
    $v_{\mathbb{F}}(r,2r) = \rho_{\mathbb{F}}(r)+2$, and
    \smallskip
    \item[(b)]
    any $\alpha$-$\operatorname{EI}_{\mathbb{F}}(2r,r,n)$ with $\alpha \neq 1$ and $n = \rho_{\mathbb{F}}(r)+2$ is an $\operatorname{EITFF}_{\mathbb{F}}(2r,r,n)$.
    \end{itemize}
\end{proposition}

Our main result in this section is as follows.

\begin{theorem}
\label{thm:RHdimKn=n}
~
    \begin{itemize}
    \item[(a)]
    For any $r$, there exists an $\operatorname{EITFF}_{\mathbb{C}}(2r,r,n)$ with $\operatorname{dim} \mathcal{K}_n = n = \rho_{\mathbb{C}}(r)+2$.

    \smallskip
    
    \item[(b)]
    If $r$ is a power of~2, then every $\operatorname{EITFF}_{\mathbb{F}}(2r,r,n)$ with $n = \rho_{\mathbb{F}}(r)+2$ satisfies $\operatorname{dim} \mathcal{K}_n = n$.
    Here, $\mathbb{F}$ can be either $\mathbb{R}$ or $\mathbb{C}$.
    \end{itemize}
\end{theorem}

A proof is forthcoming.
This result may seem surprising since each inequality $n \leq \dim \mathcal{K}_j$ is a refinement of Gerzon's bound, which gives the poor estimate $v_{\mathbb{F}}(r,2r) \leq \dim \mathbb{F}_H^{2r \times 2r}$.
Indeed, $v_{\mathbb{F}}(r,2r) = \rho_{\mathbb{F}}(r)+2 = O(\log r)$, while $\dim \mathbb{F}_H^{2r \times 2r} = \Theta(r^2)$.
As mentioned in~\cite{FGLI:25}, this discrepancy might explain why Lemmens and Seidel expressed dissatisfaction with the bound $v_{\mathbb{F}}(r,d) \leq \dim \mathcal{K}_1 = \dim \mathbb{F}_H^{d\times d} - \dim \mathbb{F}_H^{r \times r} + 1$.
For $d=2r$, our bound $n \leq \dim \mathcal{K}_3 = \Theta(r^2)$ is similarly weak.
In contrast, the ultimate refinement $n \leq \dim \mathcal{K}_n$ saturates infinitely often.
Furthermore, this occurs in both the real and complex settings, which may seem surprising when compared with the conjectured performance of Gerzon's bound for real equiangular lines (Example~\ref{ex: lines}).

To prove Theorem~\ref{thm:RHdimKn=n}, we rely on the characterization of EITFFs with $d=2r$ from~\cite{FGLI:25}, which uses the following terminology.
A nonempty subset $\mathcal{R}$ of $\mathbb{F}^{r\times r}$ is called a \textbf{$\rho$-space} if every $A \in \mathcal{R}$ is a scalar multiple of a unitary and $\mathcal{R}$ is a real vector space (that is, closed under addition and real scalar multiplication).
Then one can show that any $A,B \in \mathcal{R}$ satisfy
\[
A^*B + B^*A = 2z I
\]
for
\[
z = \langle A, B \rangle_{\rho} 
:= \tfrac{1}{r} \operatorname{Re} \langle A, B \rangle_F
= \tfrac{1}{r} \operatorname{Re} \operatorname{tr}(A B^*).
\]
In the following, we treat any $\rho$-space as a real Hilbert space with inner product $\langle \cdot, \cdot \rangle_\rho$.
A sequence $\{ C_j \}_{j=1}^n$ in $\mathbb{F}^{r\times r}$ is called \textbf{$\rho$-orthonormal} if each $C_j$ is a unitary and
\[
C_i^* C_j + C_j^* C_i = 0
\quad
\text{whenever $i \neq j$}.
\]
Clearly any orthonormal basis for a $\rho$-space is $\rho$-orthonormal; conversely, the real span of any $\rho$-orthonormal sequence is a $\rho$-space.
Classically, the Radon--Hurwitz number $\rho_{\mathbb{F}}(r)$ was defined as the maximum length of a $\rho$-orthonormal sequence in $\mathbb{F}^{r\times r}$, that is, the maximum dimension of a $\rho$-space in $\mathbb{F}^{r\times r}$.
Non-obviously, this number is given by~\eqref{eq: radon hurwitz number values}; see~\cite{Radon:22,Hurwitz:22,Adams:62,AdamsLP:65,Rajwade:93,FGLI:25}.
Finally, a sequence $\{ B_j \}_{j=1}^m$ in a $\rho$-space $\mathcal{R}$ is called a \textbf{simplex} if each $B_j$ is a unitary, and $\langle B_i, B_j \rangle_\rho = -\tfrac{1}{m-1}$ for $i \neq j$.
In that case, the real span of $\{ B_j \}_{j=1}^m$ has dimension $m-1$, as seen from the rank of $\left[ \langle B_i, B_j \rangle_\rho \right]_{i,j=1}^m$.

\begin{proposition}[Theorems 3.1 and 3.2(a) of \cite{FGLI:25}]
\label{prop:RHform}
Let $r \geq 1$ and $n \geq 3$ be integers.
Then the following hold.
\smallskip

    \begin{itemize}

    \item[(a)]
    Any $\operatorname{EITFF}_\mathbb{F}(2r,r,n)$ is equivalent to an $\operatorname{EITFF}$ with isometries of the form 
    \begin{equation}
    \label{eq:rhnorm}
    \Phi_1 = \begin{bmatrix}
        I_r \\
        0
    \end{bmatrix}, \quad \Phi_j = \begin{bmatrix}
    \alpha I_r \\
    \beta B_j
    \end{bmatrix}, \quad j > 1, 
    \end{equation}
    where 
    \begin{equation}
    \label{eq:alpha for RH form}
    \alpha := \sqrt{\frac{n-2}{2n-2}}, \quad \beta := \sqrt{1 - \alpha^2},
    \end{equation}
     and $\{B_j\}_{j=2}^n$ are unitaries in $\mathbb{F}^{r \times r}$ that satisfy: 
    \begin{equation}
    \label{eq:simplex BiBj relation}
    B_j^*B_i + B_i^*B_j 
    = \frac{-2}{n-2}I_r, \quad \forall i,j \in \{2, \dots n\}, i \neq j.
    \end{equation}

    \item[(b)]
    For any matrices $\{B_j\}_{j=2}^n$ that are unitary and satisfy~\eqref{eq:simplex BiBj relation}, defining $\{\Phi_j\}_{j=1}^n$ by \eqref{eq:rhnorm} and \eqref{eq:alpha for RH form} yields the isometries of an $\operatorname{EITFF}_\mathbb{F}(2r,r,n)$.

    \medskip

    \item[(c)]
    Matrices $\{B_j\}_{j=2}^n$ are unitary and satisfy \eqref{eq:simplex BiBj relation} if and only if they form a simplex in a $\rho$-space.

    \end{itemize}
\end{proposition}

Our proof of Theorem~\ref{thm:RHdimKn=n} relies on the next two lemmas.

\begin{lemma}
\label{lem:Sform}
Let $r \geq 1$ and $n \geq 3$ be integers.
Given an $\operatorname{EITFF}_{\mathbb{F}}(2r,r,n)$ with isometries as in Proposition~\ref{prop:RHform}(a), choose any $\rho$-orthonormal basis $\{C_j\}_{j=1}^{n-2}$ for the $\rho$-space spanned by $\{ B_j : j > 1 \}$.
Then
\[
\dim \mathcal{K}_n = n + \dim \{ C_0 \in \mathbb{F}^{r \times r} : C_0^*C_j + C_j^*C_0 = 0 \text{ for every } j > 0 \},
\]
where the set on the right-hand side is a real vector space.
\end{lemma}

\begin{proof}
Select any isometries $\{\Phi_j\}_{j=1}^n$ of the form~\eqref{eq:rhnorm}.
That is,
$\Phi_1 = \left[\begin{smallmatrix}
    I_r \\
    0
\end{smallmatrix}\right]$ and $\Phi_j = \left[\begin{smallmatrix}
\alpha I_r \\
\beta B_j
\end{smallmatrix}\right]$ for $j > 1$, where $\alpha$ and $\beta$ are as in~\eqref{eq:alpha for RH form}, and where $\{B_j\}_{j=2}^n$ forms a simplex in a $\rho$-space, and so the real span $\mathcal{R}$ of $\{B_j : j > 1 \}$ is a $\rho$-space of dimension $n-2$.
Choose any $\rho$-orthonormal basis $\{C_j\}_{j=1}^{n-2}$ for $\mathcal{R}$.

Recall from Proposition~\ref{prop: independence} that the orthogonal projections $P_j = \Phi_j\Phi_j^*$ form a linearly independent sequence in
\begin{align*}
\mathcal{K}_n 
&= \bigcap_{j=1}^n \big\{M \in \mathbb{F}_H^{2r \times 2r}: P_j M P_j \in \operatorname{span}\{P_j\} \big\} \\
&= \bigcap_{j=1}^n \big\{M \in \mathbb{F}_H^{2r \times 2r}: \Phi_j^*M\Phi_j \in \operatorname{span} \{I_r\} \big\}.
\end{align*}
We prove that
\[
\operatorname{span}\{P_1, \dots P_n\}^\perp \cap \mathcal{K}_n = \biggl\{ \begin{bmatrix} 0 & C_0^* \\ C_0 & 0 \end{bmatrix} : C_0 \in \mathbb{F}^{r \times r},\, C_0^* C_j + C_j^* C_0 = 0\; \forall j > 0 \biggr\},
\]
and the lemma follows.

For one containment, select any $C_0 \in \mathbb{F}^{r \times r}$ with $C_0^* C_j + C_j^* C_0 = 0$ for all $j > 0$, and define
\[
S:= \begin{bmatrix} 0 & C_0^* \\ C_0 & 0 \end{bmatrix} \in \mathbb{F}^{2r \times 2r}.
\]
Then $C_0^* B + B^* C_0 = 0$ for any $B \in \mathcal{R}$, so for any $j > 1$,
\[
\Phi_j^* S \Phi_j
= \begin{bmatrix} \alpha I_r & \beta B_j^* \end{bmatrix}
\begin{bmatrix} 0 & C_0^* \\ C_0 & 0 \end{bmatrix}
\begin{bmatrix} \alpha I_r \\ \beta B_j \end{bmatrix}
= \alpha \beta B_j^* C_0 + \alpha \beta C_0^* B_j
= 0.
\]
Likewise,
\[
\Phi_1^* S \Phi_1
= \begin{bmatrix} I_r & 0 \end{bmatrix}
\begin{bmatrix} 0 & C_0^* \\ C_0 & 0 \end{bmatrix}
\begin{bmatrix} I_r \\ 0 \end{bmatrix}
= 0.
\]
Consequently, $S \in \mathcal{K}_n$.
Furthermore, $S \in \operatorname{span}\{P_1,\ldots,P_n\}^\perp$ since
\[
\langle S, P_j \rangle_F 
= \operatorname{tr}(S \Phi_j \Phi_j^*)
= \operatorname{tr}(\Phi_j^* S \Phi_j)
= 0
\]
for each $j >0$.

For the other containment, select any $S \in \operatorname{span}\{P_1, \dots P_n\}^\perp \cap \mathcal{K}_n$.
To begin, we show that $\Phi_j^*S \Phi_j = 0$ for each $j > 0$.
Indeed, since $S \in \mathcal{K}_n$, there exists $c_j \in \mathbb{R}$ such that $\Phi_j^*S \Phi_j = c_jI_r$, and since $S \in \operatorname{span}\{P_1,\ldots,P_n\}^\perp$,
\[
0 = \langle S, P_j \rangle_F
= \operatorname{tr}(\Phi_j^* S \Phi_j)
= \operatorname{tr}(c_jI_r)
= rc_j.
\]
The claim follows.
Applying the claim with $j = 1$, we find
\[
S 
= \begin{bmatrix}
    0 & C_0^* \\
    C_0 & Z
\end{bmatrix}
\]
for some $C_0 \in \mathbb{F}^{r \times r}$ and $Z \in \mathbb{F}_H^{r \times r}$.
Applying the claim with $j > 1$, we find
\[
0 = \Phi_j^*S\Phi_j
= \begin{bmatrix}
    \alpha I_r & \beta B_j^*
\end{bmatrix} 
\begin{bmatrix}
    0 & C_0^* \\
    C_0 & Z
\end{bmatrix}
\begin{bmatrix}
    \alpha I_r \\
    \beta B_j
\end{bmatrix} 
= \alpha \beta B_j^*C_0 + \alpha \beta C_0^*B_j + \beta^2 B_j^*ZB_j.
\]
After multiplying on the left by the unitary $B_j$ and on the right by $B_j^*$, this equation rearranges to show
\begin{equation}
\label{eq: Z}
Z = 
\tfrac{- \alpha}{\beta} (B_j C_0^* + C_0 B_j^*)
\quad
\text{for every }j > 1.
\end{equation}
Adding over all $j \in \{2, \ldots, n\}$, we find
\[
(n-1)Z 
= \tfrac{- \alpha}{\beta} \sum_{j=2}^{n} (B_j C_0^* + C_0 B_j^*) 
= \biggl(\tfrac{- \alpha}{\beta} \sum_{j=2}^{n}B_j \biggr) C_0^*  + C_0 \biggl(\tfrac{- \alpha}{\beta} \sum_{j=2}^{n}B_j^* \biggr).
\]

Next, we show that $\sum_{j=2}^n B_j = 0$.
Recall that $\langle B_j, B_i \rangle_\rho = - \tfrac{1}{n-2}$ when $i \neq j$, while $\langle B_j, B_j \rangle = 1$ since $B_j$ is unitary.
Consequently, the all-ones vector lies in the kernel of $\begin{bmatrix} \langle B_j, B_i \rangle_\rho \end{bmatrix}_{i,j=2}^n$, and the claim follows by a standard argument.
Taking an adjoint, we find $\sum_{j=2}^n B_j^* = 0$, so that $Z=0$.
Furthermore,~\eqref{eq: Z} now reads
\[
B_j C_0^* + C_0 B_j^* = 0
\quad
\text{for every } j > 1.
\]
For each $j > 0$, $C_j$ can be expressed as a real linear combination of $B_2,\ldots,B_n$, and it follows that
\[
C_j C_0^* + C_0 C_j^* = 0
\quad
\text{for every } j > 0.
\]
Finally, multiplying on the left by the unitary $C_j^*$ and on the right by $C_j$ gives
\[
C_0^*C_j + C_j^*C_0 = 0
\quad
\text{for every } j > 0,
\]
as desired.
\end{proof}

\begin{lemma}
\label{lem: Kay's reduction}
Let $r$ be a power of~2, and suppose $\{ C_j \}_{j=1}^m$ is a $\rho$-orthonormal sequence in $\mathbb{C}^{r \times r}$ with $m \geq \rho_{\mathbb{C}}(r) - 2$.
If $C_0 \in \mathbb{C}^{r \times r}$ satisfies
\[
C_0^* C_j + C_j^*C_0 = 0
\quad
\text{for every }
j >0,
\]
then $C_0$ is a scalar multiple of a unitary (possibly zero).
\end{lemma}

\begin{proof}
We may assume $r \geq 2$ since every $1 \times 1$ matrix is a scalar multiple of a unitary; then $m \geq \rho_{\mathbb{C}}(r) - 2 \geq 2$.
For another reduction, observe the following.
Given any unitaries $U,V \in \mathbb{C}^{r \times r}$, define $\tilde{C}_j:=U^* C_j V$ for every $j \geq 0$.
Then it is easy to show that $\{ \tilde{C}_j \}_{j=1}^m$ is $\rho$-orthonormal, that
\[
\tilde{C}_0^* \tilde{C}_j + \tilde{C}_j^* \tilde{C}_0 = 0
\quad
\text{for every }
j > 0,
\]
and that $\tilde{C}_0$ is a scalar multiple of a unitary if and only if $C_0$ is a scalar multiple of a unitary.
Hence, we may replace $C_j$ by $\tilde{C}_j$ without loss of generality.

We now perform this reduction twice.
First, by taking $U = C_m$, we may assume that $C_m = I$.
Then
\[
C_j^* + C_j = C_j^* C_m + C_m^* C_j = 0
\quad
\text{for every }
j \leq m-1,
\]
so that $C_j^* = -C_j$ whenever $j \leq m-1$.
In particular, the skew-Hermitian unitary $C_{m-1}$ has eigenvalues in $\{ \pm \mathrm{i} \}$.
Second, by selecting $U=V$ to be a unitary that diagonalizes $C_{m-1}$, we may assume that $C_m = I$ and that
\[
C_{m-1}
= 
\begin{bmatrix}
\mathrm{i} I_\ell & 0 \\
0 & -\mathrm{i} I_{r-\ell}
\end{bmatrix}
\]
for some $\ell \in \{0,\ldots, r \}$.
It continues to hold that $C_j^* = - C_j$ whenever $j \leq m-1$.

With these reductions in place, write $r = 2^k$ and induct on $k \geq 1$.
In the base case $k=1$, the skew-Hermitian matrix $C_0 \in \mathbb{C}^{2 \times 2}$ takes the form
\[
C_0 = \begin{bmatrix}
a \mathrm{i} & -\overline{z} \\
z & b \mathrm{i}
\end{bmatrix}
\]
for some $a,b\in \mathbb{R}$ and $z \in \mathbb{C}$, while
\[
C_{m-1} = \begin{bmatrix}
\varepsilon \mathrm{i} & 0 \\
0 & \eta \mathrm{i}
\end{bmatrix}
\]
for some $\varepsilon,\eta \in \{ \pm 1 \}$.
Since
\[
0 = C_0^* C_{m-1} + C_{m-1}^* C_0
= \begin{bmatrix}
2 \varepsilon a & (\varepsilon + \eta) \mathrm{i} \overline{z} \\
-(\varepsilon + \eta) \mathrm{i} z & 2 \eta b
\end{bmatrix},
\]
$a = b = 0$, and so $C_0$ is a scalar multiple of a unitary.

Now assume $k \geq 2$, so that 
\[
m \geq \rho_{\mathbb{C}}(r) - 2 = 2k \geq 4.
\]
Then our reductions above have further consequences.
Namely, for every $j \leq m-2$, the skew-Hermitian matrix $C_j$ can be written as
\[
C_j = \begin{bmatrix}
D_j & -E_j^* \\
E_j & F_j
\end{bmatrix}
\]
for some $D_j \in \mathbb{C}^{\ell \times \ell}$, $E_j \in \mathbb{C}^{(r - \ell) \times \ell}$, and $F_j \in \mathbb{C}^{(r - \ell) \times (r - \ell)}$ with $D_j^* = -D_j$ and $F_j^* = -F_j$.
Then the relation
\[
0 = C_j^* C_{m-1} + C_{m-1}^* C_j 
= \begin{bmatrix}
-2 \mathrm{i} D_j & 0 \\
0 & 2 \mathrm{i} F_j
\end{bmatrix}
\]
shows that $D_j = 0$ and $F_j=0$, so that
\[
C_j =
\begin{bmatrix}
0 & -E_j^* \\
E_j & 0
\end{bmatrix}.
\]

Since $m \geq 4$, $C_{m-2}$ is invertible.
It takes the form above, so $E_{m-2}$ must be square.
In other words, $\ell = r/2 = 2^{k-1}$.
Next, since $C_j$ is unitary whenever $1 \leq j \leq m-2$, $\{E_j\}_{j=1}^{m-2}$ is a sequence of unitaries.
Furthermore, for any $i,j \leq m-2$ with $i \neq j$, it holds that
\[
0 = C_i^* C_j + C_j^* C_i
= \begin{bmatrix}
E_i^* E_j + E_j^* E_i & 0 \\
0 & -E_i E_j^* - E_j E_i^*
\end{bmatrix}.
\]
Consequently, $\{E_j\}_{j=1}^{m-2}$ is a $\rho$-orthonormal sequence in $\mathbb{C}^{\tfrac{r}{2} \times \tfrac{r}{2}}$, where 
\[
m-2 \geq \rho_{\mathbb{C}}(r)-4
= \rho_{\mathbb{C}}(r/2)-2.
\]
Applying the inductive hypothesis, we find $E_0$ is a scalar multiple of a unitary.
The same then holds for $C_0$.
\end{proof}

We now prove this section's main result.

\begin{proof}[Proof of Theorem~\ref{thm:RHdimKn=n}]
To prove (a), we show there is a $\rho$-orthonormal sequence $\{C_j\}_{j=1}^{\rho_{\mathbb{C}}(r)}$ in $\mathbb{C}^{r \times r}$ such that 
\begin{equation}
\label{eq: C0 property}
\{ C_0 \in \mathbb{C}^{r \times r} : C_0^* C_j + C_j^* C_0 = 0 \text{ for every }j > 0 \}
= \{0\}.
\end{equation}
Assuming for the moment that this holds, select any simplex $\{B_j\}_{j=1}^{\rho_{\mathbb{C}}(r)+1}$ in the $\rho$-space spanned by this $\rho$-orthonormal sequence.
(Such a simplex can be constructed, for example, as in Theorem~3.2(a) of~\cite{FGLI:25}.)
Then the $\operatorname{EITFF}_{\mathbb{C}}(2r,r,n)$ constructed in Proposition~\ref{prop:RHform}(b) has $\dim \mathcal{K}_n = n = \rho_{\mathbb{F}}(r)+2$, by Lemma~\ref{lem:Sform}.

To prove the claim, write $r = (2a+1)2^k$ for some integers $a,k \geq 0$.
Holding $a$ fixed, we induct on $k$.
Here, $\rho_{\mathbb{C}}(r) = 2k+2$.
In the base case $k=0$, define $C_1 := \mathrm{i}I_r$ and $C_2 := I_r$.
Given $C_0 \in \mathbb{C}^{r \times r}$ with $C_0^* C_j + C_j^* C_0 = 0$ for all $j > 0$, we have $C_0^* = C_0$ (taking $j = 1$) and $C_0^* = -C_0$ (taking $j = 2$), so that $C_0 = 0$.

For the inductive step, let $k \geq 0$ be arbitrary, put $r := (2a+1)2^k$, and suppose $\{C_j\}_{j=1}^{\rho_{\mathbb{C}}(r)}$ is a $\rho$-orthonormal sequence in $\mathbb{C}^{r \times r}$ for which~\eqref{eq: C0 property} holds.
We produce such a $\rho$-orthonormal sequence in $\mathbb{C}^{2r \times 2r}$, where $\rho_{\mathbb{C}}(2r) = \rho_{\mathbb{C}}(r)+2$.
First, define
\[
D_j 
:= \begin{bmatrix}
    0 & -C_j^* \\
    C_j & 0
\end{bmatrix},
\quad
1 \leq j \leq \rho_{\mathbb{C}}(r).
\]
Next, define
\[
D_{\rho_{\mathbb{C}}(r)+1} :=
\begin{bmatrix}
\mathrm{i} I_r & 0 \\
0 & -\mathrm{i} I_r
\end{bmatrix}
\qquad
\text{and}
\qquad
D_{\rho_{\mathbb{C}}(r)+2} := 
\begin{bmatrix}
I_r & 0 \\
0 & I_r
\end{bmatrix}.
\]
Then it is easy to verify that $\{ D_j \}_{j=1}^{\rho_{\mathbb{C}}(2r)}$ is $\rho$-orthonormal.
To see it has the desired property, suppose $D_0 \in \mathbb{C}^{2r \times 2r}$ satisfies
\[
D_0^* D_j + D_j^* D_0 = 0
\quad
\text{ for all }
j > 0.
\]
Taking $j = \rho_{\mathbb{C}}(r)+2$, we find $D_0^* = -D_0$.
Consequently,
\[
D_0 =
\begin{bmatrix}
A & -C_0^* \\
C_0 & B
\end{bmatrix}
\]
for some $A,B,C_0 \in \mathbb{C}^{r \times r}$ with $A^* = -A$ and $B^* = -B$.
Next, taking $j = \rho_{\mathbb{C}}(r)+1$, we find
\[
0 = D_0^* D_{\rho_{\mathbb{C}}(r)+1} + D_{\rho_{\mathbb{C}}(r)+1}^* D_0
= \begin{bmatrix}
-2\mathrm{i} A & 0 \\
0 & 2 \mathrm{i} B
\end{bmatrix},
\]
so that $A=B=0$ and
\[
D_0 =
\begin{bmatrix}
0 & -C_0^* \\
C_0 & 0
\end{bmatrix}.
\]
Finally, taking any $j \leq \rho_{\mathbb{C}}(r)$, we find
\[
0 = D_0^* D_j + D_j^* D_0
= \begin{bmatrix}
C_0^* C_j + C_j^* C_0 & 0 \\
0 & C_0C_j^* + C_j C_0^*
\end{bmatrix}.
\]
Thus, $C_0^* C_j + C_j^* C_0 = 0$ for all $j > 0$, and $C_0 = 0$ by the inductive hypothesis.
It follows that $D_0 = 0$.
This completes the proof of (a).

For (b), let $r$ be any power of~2, and pick either $\mathbb{F} \in \{\mathbb{R},\mathbb{C}\}$.
By Lemma \ref{lem:Sform}, it suffices to show that, for any $\rho$-orthonormal sequence $\{C_j \}_{j=1}^{\rho_{\mathbb{F}}(r)}$ in $\mathbb{F}^{r \times r}$,
\[
\{ C_0 \in \mathbb{F}^{r \times r} : C_0^* C_j + C_j^* C_0 = 0 \text{ for every }j>0\} 
= \{0\}.
\]
To accomplish this, we apply Lemma~\ref{lem: Kay's reduction}.
Indeed, by considering each of the four cases $c \in \{0,1,2,3\}$ in~\eqref{eq: radon hurwitz number values} separately, we find $\rho_{\mathbb{F}}(r) \geq \rho_{\mathbb{C}}(r) - 2$.
Now suppose $C_0 \in \mathbb{F}^{r \times r}$ satisfies
\[
C_0^* C_j + C_j^* C_0 = 0
\quad
\text{for every }
j>0.
\]
By Lemma~\ref{lem: Kay's reduction}, $C_0$ is a scalar multiple of a unitary.
If $C_0$ were not zero, then defining $C_{\rho_{\mathbb{F}}(r)+1}:=\|C_0\|_\rho^{-1}C_0$ would produce a $\rho$-orthonormal sequence $\{C_j\}_{j=1}^{\rho_{\mathbb{F}}(r)+1}$ of length $\rho_{\mathbb{F}}(r)+1$ in $\mathbb{F}^{r \times r}$.
This is impossible, so $C_0 = 0$, as desired.
\end{proof}

\begin{remark}
Theorem~\ref{thm:RHdimKn=n}(b) fails when $r$ is not a power of~2, since there exists an $\operatorname{EITFF}_\mathbb{F}(2r,r,n)$ with  $\dim \mathcal{K}_n \neq n = \rho_\mathbb{F}(r)+2$.
To see this, first write $r = (2a+1)2^k$ for integers $a \geq 1$ and $k \geq 0$, so that $\rho_\mathbb{F}(r) = \rho_\mathbb{F}(2^k)$.
As in the proof of Lemma~\ref{lem: Kay's reduction}, there exists a $\rho$-orthonormal sequence $\{A_j\}_{j=1}^{\rho_\mathbb{F}(r)}$ in $\mathbb{F}^{2^k \times 2^k}$ with $A_{\rho_\mathbb{F}(r)} = I$.
Put
\[
D := 
\begin{bmatrix}
    1 & 0 \\
    0 & -I_{2a}
\end{bmatrix}
\in \mathbb{F}^{(2a+1) \times (2a+1)},
\]
and define $\{C_j\}_{j=1}^{\rho_\mathbb{F}(r)}$ by $C_j := D \otimes A_j \in \mathbb{F}^{r \times r}$ when $1 \leq j < \rho_\mathbb{F}(2^k)$ and $C_{\rho_\mathbb{F}(2^k)} := I_{r}$.
Then it is easy to check that $\{C_j\}_{j=1}^{\rho_\mathbb{F}(r)}$ is a $\rho$-orthonormal sequence, and that the nonzero matrix
\[
C_0 
:= \begin{bmatrix}
    0 & I_{2^k} & 0 \\
    -I_{2^k} & 0 & 0 \\
   0 & 0 & 0
\end{bmatrix} \in \mathbb{F}^{r \times r}
\]
satisfies $C_0^*C_j +C_j^*C_0$ for every $j >0$.
Next, select $\{B_j\}_{j=1}^{\rho_\mathbb{F}(r)+1}$ to be any simplex in the $\rho$-space spanned by $C_1,\ldots,C_{\rho_\mathbb{F}(r)}$, and construct an $\operatorname{EITFF}_{\mathbb{F}}(2r,r,\rho_\mathbb{F}(r)+2)$ as in Proposition~\ref{prop:RHform}.
By Lemma~\ref{lem:Sform}, this EITFF has $\dim \mathcal{K}_n > n$.

Similarly, Theorem~\ref{thm:RHdimKn=n}(a) fails when $\mathbb{C}$ is replaced by $\mathbb{R}$: when $r > 1$ is odd, one can show that every $\operatorname{EITFF}_\mathbb{R}(2r,r,\rho_\mathbb{R}+2)$ has $\dim \mathcal{K}_n > n$.
Indeed, $\rho_{\mathbb{R}}(r) = 1$ and for any unitary $C_1 \in \mathbb{R}^{r \times r}$ there exists $C_0 \neq 0$ with $C_0^* C_1 + C_1^* C_0 = 0$, namely $C_0 = A C_1$ for any choice of $A = -A^* \neq 0$.
The claim follows by Lemma~\ref{lem:Sform}.

However, there is room to extend Theorem~\ref{thm:RHdimKn=n}(a) over the reals, since there exists $r \neq 2^k$ for which some $\operatorname{EITFF}_{\mathbb{R}}(2r,r,n)$ has $\dim \mathcal{K}_n = n = \rho_{\mathbb{R}}(r)+2$.
Writing $r = (2a+1)2^k$ as above, for each $a \in \{1,2\}$ and each $k \in \{2,3\}$ there exists a $\rho$-orthonormal sequence $\{ C_j \}_{j=1}^{\rho_{\mathbb{R}}(r)}$ in $\mathbb{R}^{r \times r}$ for which a computer calculation shows the only $C_0 \in \mathbb{R}^{r \times r}$ with $C_0^* C_j + C_j^* C_0 = 0$ for every $j > 0$ is $C_0 = 0$.
(Specifically, one can take $C_j = A_j \otimes I_{2a+1}$, where $A_{\rho_{\mathbb{R}}(2^k)} = I_{2^k}$ and $\{A_j\}_{j=1}^{\rho_{\mathbb{R}}(2^k)-1}$ is a $\rho$-orthonormal sequence of skew-Hermitian unitaries that appears on page~3207 of~\cite{FGLI:25}.)
Beyond these four examples and powers of~2, the authors do not know for which even $r$ there exists an $\operatorname{EITFF}_{\mathbb{R}}(2r,r,n)$ with $\dim \mathcal{K}_n = n = \rho_{\mathbb{R}}(r)+2$, and we leave this problem for future research.
\end{remark}


\section*{Acknowledgments}

JWI was supported by NSF DMS 2220301 and a grant from the Simons Foundation.
The authors are pleased to thank Dustin Mixon for comments that led to an improved exposition in Section~\ref{sec: lower bound}.

\bigskip

\bibliographystyle{abbrv}
\bibliography{refs}

\end{document}